\ifx\onecol\undefined
\documentclass[journal]{IEEEtran}
\else 
\documentclass[onecolumn,draftcls,12pt]{IEEEtran}
\fi

\usepackage[utf8]{inputenc}
\usepackage{physics}
\usepackage{amsmath}
\usepackage{amssymb}
\usepackage{subfigure}
\usepackage{graphicx}
\usepackage{graphics}
\usepackage{color}
\usepackage{psfrag}
\usepackage{cite}
\usepackage{balance}
\usepackage{algorithm}
\usepackage{accents}
\usepackage{amsthm}
\usepackage{bm}
\usepackage{url}
\usepackage{algorithmic}
\usepackage[english]{babel}
\usepackage{multirow}
\usepackage{enumerate}
\usepackage{cases}
\usepackage{stfloats}
\usepackage{dsfont}
\usepackage{soul}
\usepackage{amsfonts}
\usepackage{fancyhdr}
\usepackage{hhline}
\usepackage{array}
\usepackage{booktabs}
\usepackage{framed}
\usepackage{float}
\usepackage{soul}
\usepackage{url}
\newcommand{\mb}[1]{{  \mathbf  #1}}  %\mathbf  \bm

\begin{document}
\newtheorem{theorem}{Theorem}
\newtheorem{acknowledgement}[theorem]{Acknowledgement}
\newtheorem{axiom}[theorem]{Axiom}
\newtheorem{case}[theorem]{Case}
\newtheorem{claim}[theorem]{Claim}
\newtheorem{conclusion}[theorem]{Conclusion}
\newtheorem{condition}[theorem]{Condition}
\newtheorem{conjecture}[theorem]{Conjecture}
\newtheorem{criterion}[theorem]{Criterion}
\newtheorem{definition}{Definition}
\newtheorem{exercise}[theorem]{Exercise}
\newtheorem{lemma}{Lemma}
\newtheorem{corollary}{Corollary}
\newtheorem{notation}[theorem]{Notation}
\newtheorem{problem}[theorem]{Problem}
\newtheorem{proposition}{Proposition}
\newtheorem{solution}[theorem]{Solution}
\newtheorem{summary}[theorem]{Summary}
\newtheorem{assumption}{Assumption}
\newtheorem{example}{\bf Example}
\newtheorem{remark}{\bf Remark}

\newtheorem{thm}{Corollary}[section]
\renewcommand{\thethm}{\arabic{section}.\arabic{thm}}

\def\qed{$\Box$}
\def\QED{\mbox{\phantom{m}}\nolinebreak\hfill$\,\Box$}
\def\proof{\noindent{\emph{Proof:} }}
\def\poof{\noindent{\emph{Sketch of Proof:} }}
\def
\endproof{\hspace*{\fill}~\qed
\par
\endtrivlist\unskip}
\def\endproof{\hspace*{\fill}~\qed\par\endtrivlist\vskip3pt}

\def\E{\mathsf{E}}
\def\eps{\varepsilon}
\def\phi{\varphi}
\def\Lsp{{\boldsymbol L}}
\def\Bsp{{\boldsymbol B}}
\def\lsp{{\boldsymbol\ell}}
\def\Ltsp{{\Lsp^2}}
\def\Lpsp{{\Lsp^p}}
\def\Linsp{{\Lsp^{\infty}}}
\def\LtR{{\Lsp^2(\Rst)}}
\def\ltZ{{\lsp^2(\Zst)}}
\def\ltsp{{\lsp^2}}
\def\ltZt{{\lsp^2(\Zst^{2})}}
\def\ninN{{n{\in}\Nst}}
\def\oh{{\frac{1}{2}}}
\def\grass{{\cal G}}
\def\ord{{\cal O}}
\def\dist{{d_G}}
\def\conj#1{{\overline#1}}
\def\ntoinf{{n \rightarrow \infty}}
\def\toinf{{\rightarrow \infty}}
\def\tozero{{\rightarrow 0}}
\def\trace{{\operatorname{trace}}}
\def\ord{{\cal O}}
\def\UU{{\cal U}}
\def\rank{{\operatorname{rank}}}
\def\acos{{\operatorname{acos}}}

\def\SINR{\mathsf{SINR}}
\def\SNR{\mathsf{SNR}}
\def\SIR{\mathsf{SIR}}
\def\tSIR{\widetilde{\mathsf{SIR}}}
\def\Ei{\mathsf{Ei}}
\def\l{\left}
\def\r{\right}
\def\lb{\left\{}
\def\rb{\right\}}

\setcounter{page}{1}

% Definitions
\newcommand{\eref}[1]{(\ref{#1})}
\newcommand{\fig}[1]{Fig.\ \ref{#1}}

% Bold lowercase
\def\bydef{:=}
\def\ba{{\mathbf{a}}}
\def\bb{{\mathbf{b}}}
\def\bc{{\mathbf{c}}}
\def\bd{{\mathbf{d}}}
\def\bee{{\mathbf{e}}}
\def\bff{{\mathbf{f}}}
\def\bg{{\mathbf{g}}}
\def\bh{{\mathbf{h}}}
\def\bi{{\mathbf{i}}}
\def\bj{{\mathbf{j}}}
\def\bk{{\mathbf{k}}}
\def\bl{{\mathbf{l}}}
\def\bm{{\mathbf{m}}}
\def\bn{{\mathbf{n}}}
\def\bo{{\mathbf{o}}}
\def\bp{{\mathbf{p}}}
\def\bq{{\mathbf{q}}}
\def\br{{\mathbf{r}}}
\def\bs{{\mathbf{s}}}
\def\bt{{\mathbf{t}}}
\def\bu{{\mathbf{u}}}
\def\bv{{\mathbf{v}}}
\def\bw{{\mathbf{w}}}
\def\bx{{\mathbf{x}}}
\def\by{{\mathbf{y}}}
\def\bz{{\mathbf{z}}}
\def\b0{{\mathbf{0}}}

% Bold capital letters
\def\bA{{\mathbf{A}}}
\def\bB{{\mathbf{B}}}
\def\bC{{\mathbf{C}}}
\def\bD{{\mathbf{D}}}
\def\bE{{\mathbf{E}}}
\def\bF{{\mathbf{F}}}
\def\bG{{\mathbf{G}}}
\def\bH{{\mathbf{H}}}
\def\bI{{\mathbf{I}}}
\def\bJ{{\mathbf{J}}}
\def\bK{{\mathbf{K}}}
\def\bL{{\mathbf{L}}}
\def\bM{{\mathbf{M}}}
\def\bN{{\mathbf{N}}}
\def\bO{{\mathbf{O}}}
\def\bP{{\mathbf{P}}}
\def\bQ{{\mathbf{Q}}}
\def\bR{{\mathbf{R}}}
\def\bS{{\mathbf{S}}}
\def\bT{{\mathbf{T}}}
\def\bU{{\mathbf{U}}}
\def\bV{{\mathbf{V}}}
\def\bW{{\mathbf{W}}}
\def\bX{{\mathbf{X}}}
\def\bY{{\mathbf{Y}}}
\def\bZ{{\mathbf{Z}}}

\def\bxi{{\boldsymbol{\xi}}}

\def\sT{{\mathsf{T}}}
\def\sH{{\mathsf{H}}}
\def\cmp{{\text{cmp}}}
\def\cmm{{\text{cmm}}}
\def\WPT{{\text{WPT}}}
\def\lo{{\text{lo}}}
\def\gl{{\text{gl}}}

\def\tT{{\widetilde{T}}}
\def\tF{{\widetilde{F}}}
\def\tP{{\widetilde{P}}}
\def\tG{{\widetilde{G}}}
\def\tbh{{\widetilde{\mathbf{h}}}}
\def\tbg{{\widetilde{\mathbf{g}}}}

% mathbb Bold capital letters
\def\mA{{\mathbb{A}}}
\def\mB{{\mathbb{B}}}
\def\mC{{\mathbb{C}}}
\def\mD{{\mathbb{D}}}
\def\mE{{\mathbb{E}}}
\def\mF{{\mathbb{F}}}
\def\mG{{\mathbb{G}}}
\def\mH{{\mathbb{H}}}
\def\mI{{\mathbb{I}}}
\def\mJ{{\mathbb{J}}}
\def\mK{{\mathbb{K}}}
\def\mL{{\mathbb{L}}}
\def\mM{{\mathbb{M}}}
\def\mN{{\mathbb{N}}}
\def\mO{{\mathbb{O}}}
\def\mP{{\mathbb{P}}}
\def\mQ{{\mathbb{Q}}}
\def\mR{{\mathbb{R}}}
\def\mS{{\mathbb{S}}}
\def\mT{{\mathbb{T}}}
\def\mU{{\mathbb{U}}}
\def\mV{{\mathbb{V}}}
\def\mW{{\mathbb{W}}}
\def\mX{{\mathbb{X}}}
\def\mY{{\mathbb{Y}}}
\def\mZ{{\mathbb{Z}}}

% Caligraphic capital letters
\def\cA{\mathcal{A}}
\def\cB{\mathcal{B}}
\def\cC{\mathcal{C}}
\def\cD{\mathcal{D}}
\def\cE{\mathcal{E}}
\def\cF{\mathcal{F}}
\def\cG{\mathcal{G}}
\def\cH{\mathcal{H}}
\def\cI{\mathcal{I}}
\def\cJ{\mathcal{J}}
\def\cK{\mathcal{K}}
\def\cL{\mathcal{L}}
\def\cM{\mathcal{M}}
\def\cN{\mathcal{N}}
\def\cO{\mathcal{O}}
\def\cP{\mathcal{P}}
\def\cQ{\mathcal{Q}}
\def\cR{\mathcal{R}}
\def\cS{\mathcal{S}}
\def\cT{\mathcal{T}}
\def\cU{\mathcal{U}}
\def\cV{\mathcal{V}}
\def\cW{\mathcal{W}}
\def\cX{\mathcal{X}}
\def\cY{\mathcal{Y}}
\def\cZ{\mathcal{Z}}
\def\cd{\mathcal{d}}
\def\Mt{M_{t}}
\def\Mr{M_{r}}
%% my defs
\def\O{\Omega_{M_{t}}}
\newcommand{\figref}[1]{{Fig.}~\ref{#1}}
\newcommand{\tabref}[1]{{Table}~\ref{#1}}

%% From Kaibin
% \newcommand{\var}{\mathsf{Var}}
\newcommand{\fb}{\tx{fb}}
\newcommand{\nf}{\tx{nf}}
\newcommand{\BC}{\tx{(bc)}}
\newcommand{\MAC}{\tx{(mac)}}
\newcommand{\Pout}{p_{\mathsf{out}}}
\newcommand{\nnn}{\nn\\}
\newcommand{\FB}{\tx{FB}}
\newcommand{\TX}{\tx{TX}}
\newcommand{\RX}{\tx{RX}}
\renewcommand{\mod}{\tx{mod}}
\newcommand{\m}[1]{\mathbf{#1}}
\newcommand{\td}[1]{\tilde{#1}}
\newcommand{\sbf}[1]{\scriptsize{\textbf{#1}}}
\newcommand{\stxt}[1]{\scriptsize{\textrm{#1}}}
\newcommand{\suml}[2]{\sum\limits_{#1}^{#2}}
\newcommand{\sumlk}{\sum\limits_{k=0}^{K-1}}
\newcommand{\eqhsp}{\hspace{10 pt}}
\newcommand{\tx}[1]{\texttt{#1}}
\newcommand{\Hz}{\ \tx{Hz}}
\newcommand{\sinc}{\tx{sinc}}
\newcommand{\diag}{\mathrm{diag}}
\newcommand{\MAI}{\tx{MAI}}
\newcommand{\ISI}{\tx{ISI}}
\newcommand{\IBI}{\tx{IBI}}
\newcommand{\CN}{\tx{CN}}
\newcommand{\CP}{\tx{CP}}
\newcommand{\ZP}{\tx{ZP}}
\newcommand{\ZF}{\tx{ZF}}
\newcommand{\SP}{\tx{SP}}
\newcommand{\MMSE}{\tx{MMSE}}
\newcommand{\MINF}{\tx{MINF}}
\newcommand{\RC}{\tx{MP}}
\newcommand{\MBER}{\tx{MBER}}
\newcommand{\MSNR}{\tx{MSNR}}
\newcommand{\MCAP}{\tx{MCAP}}
\newcommand{\vol}{\tx{vol}}
\newcommand{\ah}{\hat{g}}
\newcommand{\tg}{\tilde{g}}
\newcommand{\teta}{\tilde{\eta}}
\newcommand{\heta}{\hat{\eta}}
\newcommand{\uh}{\m{\hat{s}}}
\newcommand{\eh}{\m{\hat{\eta}}}
\newcommand{\hv}{\m{h}}
\newcommand{\hh}{\m{\hat{h}}}
\newcommand{\Po}{P_{\mathrm{out}}}
\newcommand{\Poh}{\hat{P}_{\mathrm{out}}}
\newcommand{\Ph}{\hat{\gamma}}
\newcommand{\mat}[1]{\begin{matrix}#1\end{matrix}}
\newcommand{\ud}{^{\dagger}}
\newcommand{\C}{\mathcal{C}}
\newcommand{\nn}{\nonumber}
\newcommand{\nInf}{U\rightarrow \infty}

% \title{Breaking Information-Theoretical Limit of Rydberg Atomic MIMO Receivers via IQ-Aware Precoding}
\title{MIMO Precoding for Rydberg Atomic Receivers}
\author{{Mingyao Cui,~\IEEEmembership{Graduate Student Member, IEEE}, Qunsong Zeng,~\IEEEmembership{Member, IEEE}, and Kaibin Huang,~\IEEEmembership{Fellow, IEEE}}
\thanks{M. Cui, Q. Zeng, and K. Huang are with the Department of Electrical and Electronic Engineering, The University of Hong Kong, Hong Kong (Email: \{mycui,qszeng,huangkb\}@eee.hku.hk). Corresponding authors: Q. Zeng; K. Huang.}}

\maketitle

% =========================================================
% \thispagestyle{empty}
% \pagestyle{empty}
% =========================================================

\begin{abstract}
Leveraging the strong atom-light interaction, a Rydberg atomic receiver can measure radio waves with extreme sensitivity. Existing research primarily focuses on improving the architecture and signal detection capability of atomic receivers, while traditional signal processing schemes at the transmitter side have remained unchanged.
As a result, these schemes fail to maximize the throughput of atomic receivers, given that the coupling between atomic dipole moment and radio-wave magnitude results in a \emph{nonlinear} transmission model in contrast to the traditional linear one. 
To address this issue, our work proposes to design customized precoding techniques for atomic multiple-input-multiple-output (MIMO) systems to achieve the channel capacity.
A strong-reference approximation is initially proposed to linearize the \emph{nonlinear transition} model of atomic receivers. This facilitates the derivation of atomic-MIMO channel capacity as $\min(N_r/2, N_t)\log({\rm SNR})$ at high signal-to-noise ratios (SNRs) for $N_r$ receive atomic antennas and $N_t$ classic transmit antennas. Then, a new digital precoding technique, termed In-phase-and-Quadrature (IQ) aware precoding is presented, which features independent processing of I/Q data streams using four real-valued matrices. The design is shown to be capacity-achieving for the atomic MIMO system. In addition, for the case of large-scale MIMO system, we extend the preceding fully-digital precoding design to the popular hybrid precoding architecture, which cascades a classical analog precoder with a low-dimensional version of the proposed IQ-aware digital precoder. By alternatively optimizing the digital and analog parts, the hybrid design is able to approach the performance of the optimal IQ-aware fully digital precoding. Simulation results validate the superiority of proposed IQ-aware precoding methods over existing techniques in the context of atomic MIMO communication.
\end{abstract}

\begin{IEEEkeywords}
Atomic receivers, multiple-input-multiple-output (MIMO), precoding.
\end{IEEEkeywords}
%============================================================

\section{Introduction}
% \subsection{From Traditional Antennas to Rydberg Atomic Antennas}
Ever since Hertz validated the existence of electromagnetic (EM) waves, the accurate measurement of EM waves has become the cornerstone of modern wireless communication, remote sensing, and radar~\cite{kaushik_toward_2024, kraus_heinrich_1988}. 
    To date, a variety of antenna technologies have been developed, e.g., dipole antennas and aperture antennas.  These technologies are based on the same measurement principle: an incident EM wave drives an antenna to generate induced currents, which results in informative electrical signals. This mechanism imposes a notable limitation on the sensitivity (i.e., minimum detectable field strength) of traditional antennas. According to the fluctuation-dissipation theorem \cite{nyquist_thermal_1928}, the random movement of free electrons within metal generates indistinguishable thermal noise, known as Johnson-Nyquist noise, that contaminates electrical signals. As a result, traditional antennas typically achieve a sensitivity on the order of $1{\rm uV}/{\rm m}/\sqrt{\rm Hz}$~\cite{QSN_Bussey2022}, making ultra-precise EM-field measurement challenging to achieve \cite{QuanSense_Zhang2023, sedlacek_microwave_2012}.

Recently, at the intersection of quantum sensing and wireless communication, an emerging antenna technology known as \emph{Rydberg atomic receiver} has shown the potential in surpassing the sensitivity limitations of traditional antennas~\cite{zhang_rydberg_2024, RydMag_Fancher2021}. 
Rydberg atoms, referring to excited atoms operating at high quantum states \cite{RydReview_Saffman2010}, can strongly interact with incident EM waves, triggering electron transitions between resonant energy levels~\cite{AtomicPhysics}. Receivers based on Rydberg atoms can monitor the strength of these transitions by leveraging quantum phenomena, such as \emph{ac-Stark shift} and \emph{electromagnetically induced transparency} (EIT), to detect information carried by the fields~\cite{RydMag_Fancher2021}. 
Rydberg atomic receivers can overcome the limitations imposed by thermal noise, as the interaction between EM fields and atoms theoretically introduces no thermal noise. Furthermore, the quantum shot noise induced during the observation of Rydberg atoms' quantum states is typically orders of magnitude smaller than the thermal noise of metal antennas \cite{QuanSense_Zhang2023}. 
% With these advantages, atomic receivers has been reported to achieve measurement sensitivity on the order of $5 {\rm nV}/{\rm cm}/\sqrt{\rm Hz}$ \cite{cai_sensitivity_2023}. 
With these advantages, atomic receivers can theoretically achieve measurement sensitivity determined by standard quantum limit, e.g., $1 {\rm nV}/{\rm m}/\sqrt{\rm Hz}$ \cite{fan_atom_2015}. 
Thus, such receivers show a promise for use in long-range communication scenarios, such as satellite communications and space-air-ground integrated networks~\cite{liu_space-air-ground_2018}, where detecting weak signals is crucial. 
% In addition, they can be deployed in integrated sensing and communication systems to boost sensing accuracy~\cite{kaushik_toward_2024}.  
% and the quantum shot noise of Rydberg atoms 
% The sensitivity of atomic receivers is determined by the standard quantum limit, which can be orders of magnitude higher than the traditional counterparts~\cite{QuanSense_Zhang2023, QuanSense_Degen2017}. 
% Furthermore, atomic receivers can offer near-zero antenna spacing due to the nanometer-level atom spacing~\cite{cox_quantum-limited_2018}. Moreover, they possess a flexible operation frequency range, spanning from 100 Megahertz (MHz) to 1 Terahertz (THz), thanks to the abundant atomic energy levels~\cite{RydMag_Art2022}. 
% With these benefits, atomic receivers have the potential to replace traditional receivers in next-generation communication and sensing systems. 
Aligned with this vision, our work advocates the adoption of atomic receivers as a revolutionary approach to enhance wireless communication systems.

% \subsection{Prior Works on Rydberg Atomic Receivers}
% \subsection{Prior Works}
Existing research on the application of Rydberg atomic receivers in detecting modulated wireless signals primarily focuses on the lab demonstration of \emph{single-input-single-output} (SISO) air interfaces. 
% Specifically, through the interaction between EM fields and Rydberg atoms, the strength and frequency of EM waves are transduced into the intensity of electron transition, termed the \emph{Rabi frequency}. 
Specifically, the EM field is captured by an atomic parameter called \emph{Rabi frequency},  which characterizes the electron transition intensity. 
By inferring symbols from Rabi frequencies, researchers
have successfully demonstrated the detection of amplitude-modulated (AM) and frequency-modulated (FM) signals \cite{kumar_rydberg-atom_2017, RydAMFM_Anderson2021, yuan_rydberg_2023, RydAM_Zhen2019}. To enable phase detection, 
% researchers have employed the holographic phase-sensing method or quantum superheterodyne detection
techniques such as holographic phase-sensing or heterodyne measurement have been employed~\cite{RydPhase_And2020}. These techniques utilize a \emph{local oscillator} (LO) to send a known reference wave that interferes with the signal wave, thereby allowing phase-modulated symbols to be extracted from the phase difference between the two waves. This method has facilitated the implementation of \emph{phase-shift keying} (PSK) and \emph{quadrature-amplitude modulation} (QAM) communication systems~\cite{RydPhase_Simons2019, Rydphase_Holloway2019,RydNP_Jing2020}. 
Another representative application involves the concurrent detection of multi-frequency signals spanning GHz frequency bands \cite{RydMultiband_holloway2021, RydMultiband_Du2022, RydMultiband_Meyer2023}. This is achieved either by mixing different species of alkali atoms within a single receiver or by exciting the atoms to distinct Rydberg quantum states.

In addition to the SISO air interface, recent studies have also explored the spatial diversity and multiplexing gains of atomic antenna arrays. In \cite{AtomicSIMO_Otto2021}, a \emph{single-input-multiple-output} (SIMO) system was designed with four atomic antennas to recover AM signals. 
A notable observation from this study was that the experimental receive \emph{signal-to-noise ratio} (SNR) scales linearly with the number of atomic antennas, as in traditional systems.
This atomic SIMO configuration was later implemented for angle-of-arrival estimation using the heterodyne measurement scheme \cite{RydAOA_Robinson2021, AtomicAoA_Richardson2024}. To further reap the multiplexing gain, our prior work developed a multi-user signal detection model for atomic \emph{multiple-input-multiple-output} (MIMO) receivers \cite{AtomicMIMO_Cui2024}. We proved that atomic MIMO receivers exhibit a \emph{non-linear} input-output relationship modeled by \emph{biased phase retrieval} \cite{PR_Dong2023}, in contrast to the \emph{linear transition model} in classic MIMO systems.  Subsequently, a maximum likelihood-based algorithm was proposed for the simultaneous detection of multi-user signals.

Previous efforts have primarily concentrated on improving the receiver architectures and detection algorithms, while the 
design of a classic transmitter optimized for an atomic receiver remains largely 
uncharted. Attributed to the altered input-output relationship, existing transmitter-side signal processing methods may not be optimal for atomic receivers. 
For instance, classic MIMO communications rely on \emph{singular value decomposition} (SVD)-based precoding at the transmitter side to achieve the channel capacity $\min(N_r, N_t)\log {\rm SNR} + \mathcal{O}(1)$, where $N_r$ and $N_t$ represent
the numbers of antennas at the receiver and transmitter, respectively~\cite{capacity_telatar1999, zheng_diversity_2003}. 
However, this 
% EVD-based precoding and the associated 
capacity is based on the \emph{linear transition model}, while the atomic counterpart is non-linear. This difference raises three critical questions that remain unaddressed in the existing literature: 
\begin{itemize}
    \item \emph{What is the information-theoretical limit of an atomic MIMO system?}
    \item \emph{How can we approach this limit through transmitter-side signal processing?}
    \item \emph{Is there any loss in degree-of-freedoms (DoFs) in atomic MIMO systems due to the non-linear transition?}
\end{itemize}

Addressing these questions is pivotal for unleashing the full potential of atomic receivers for next-generation communication systems. 
To this end, this work represents the first attempt on designing the transmit precoding to approach the channel capacity of Rydberg atomic MIMO systems. 
We consider the point-to-point MIMO communications supported by the mentioned heterodyne measurement scheme~\cite{RydPhase_And2020}.
The key contributions are summarized as follows.
\begin{itemize}
    \item \textbf{Channel capacity with atomic MIMO receivers}: 
    The non-linear input-output relationship of atomic MIMO receivers poses a challenge to express its mutual information analytically. To overcome this challenge, we propose a strong-reference approximation to linearize the transition model, which is motivated by the fact that the reference wave is much stronger than the signal wave 
    % as LOs are typically deployed quite close to atomic receivers 
    due to the proximity of the LOs to atomic receivers
    \cite{RydPhase_Simons2019, Rydphase_Holloway2019, RydNP_Jing2020}. 
    Based on this approximation, we can convert the \emph{non-linear magnitude detection} model of atomic receivers into a \emph{linear real-part detection} model by discarding high-order infinitesimal quantities, which allows us to obtain the expression of channel capacity.  
    \item \textbf{IQ-aware digital precoding}: 
    An IQ-aware digital precoding is proposed to achieve the channel capacity. 
    Traditional MIMO precoding widely adopts a complex-valued precoder to process in-phase and quadrature (IQ) baseband symbols. 
    This method can only produce \emph{circularly-symmetric Gaussian} distributed signals, which are not optimal for the magnitude-detection based atomic MIMO receivers.
    In contrast, the proposed IQ-aware digital precoding employs four real-valued matrices to independently process the I/Q baseband symbols, which has the ability to produce \emph{general Gaussian} distributed signals.  
    The design is shown to achieve the channel capacity with atomic receivers, which is proven to be $\min(N_r/2, N_t)\log{\rm SNR}$ at high SNRs. This capacity result reveals that as long as the number of receive atomic antennas, $N_r$, is at least twice the number of transmit traditional antennas, $N_t$, atomic MIMO receivers experience no loss in spatial DoFs with respect to the traditional counterparts.   
    \item \textbf{IQ-aware hybrid precoding}: 
For the case of large-scale MIMO systems, we extend the preceding fully digital precoding design to the popular hybrid precoding architecture~\cite{HB_Gao}. Our design cascades the high-dimensional complex analog precoder implemented using \emph{phase shifters} (PSs)  with  a low-dimensional version of the proposed IQ-aware digital precoder, where both the fully-connected (FC) and the sub-connected (SC) PS networks are taken into account.
We exploit the alternating minimization approach to alternatively update the digital and analog precoders, with the objective of minimizing the Frobenius distance between the previous IQ-aware fully digital precoder and the current IQ-aware hybrid precoder. 
Numerical results validate their close performance.
Moreover, we show that the proposed IQ-aware hybrid precoding schemes significantly outperform classical hybrid precoding methods in the context of atomic MIMO systems. 
\end{itemize}

\emph{Organization:} 
The remainder of this paper is organized as follows. The system model of atomic MIMO receivers is provided in Section~\ref{sec:2}. In Section~\ref{sec:3}, we formulate the system channel capacity using strong-reference approximation. 
The IQ-aware digital precoding for achieving the channel capacity is discussed Section~\ref{sec:4}.
The IQ-aware hybrid precoding design for the FC and SC architectures are presented in Section~\ref{sec:5}.
In Section~\ref{sec:6}, numerical results are presented, followed by conclusions drawn in Section~\ref{sec:7}. 

\emph{Notations:} 
Lower-case and upper-case boldface letters represent
vectors and matrices, respectively. 
$\mb{X}_{m,n}$, $\mb{X}_{m,:}$, and $\mb{X}_{:, n}$ represent the $(m,n)$-th entry, the $m$-th row, and the $n$-th column of matrix $\mb{X}$, respectively. Moreover, $\mb{X}_{m:n,:}$ and $\mb{X}_{:, m:n}$ are defined as the $m$ to $n$ columns and $m$ to $n$ rows of $\mb{X}$. 
${[\cdot]^{-1}}$, ${[\cdot]^{\dag}}$, 
${[\cdot]^{*}}$, ${[\cdot]^{\rm 
T}}$, and ${[\cdot]^{\rm H}}$ denote the inverse, pseudo-inverse, conjugate, 
transpose, 
and conjugate-transpose operations, respectively. $\|\cdot\|_2$ denotes the 
$l_2$-norm of the argument. $\|\cdot\|_F$ denotes the Frobenius norm of the 
argument. 
$|\cdot|$ denotes the element-wise magnitude of its argument. 
$\mathsf{Tr}(\cdot)$ denotes the trace operator.
${\mathsf E}\left(\cdot\right)$ is the expectation operator.
$\Re(\cdot)$ denotes the real part of the argument. ${\rm arg}(\cdot)$ is the phase of the argument.
$\mathcal{CN}\left({\boldsymbol{\mu}}, {\boldsymbol \Sigma } \right)$ ($\mathcal{N}\left({\boldsymbol{\mu}}, {\boldsymbol \Sigma } \right)$) denotes the complex (real) Gaussian distribution with mean ${\boldsymbol 
	\mu}$ and covariance ${\boldsymbol \Sigma }$. ${\cal U}\left(a,b\right)$ denotes 
	the 
uniform distribution between $a$ and $b$. $\mathbf{I}_{L}$ is an $L\times L$ 
identity matrix.

\section{Review of Rydberg Atomic MIMO System}\label{sec:2}
\begin{figure}[t!]
\centering
\includegraphics[width=3.5in]{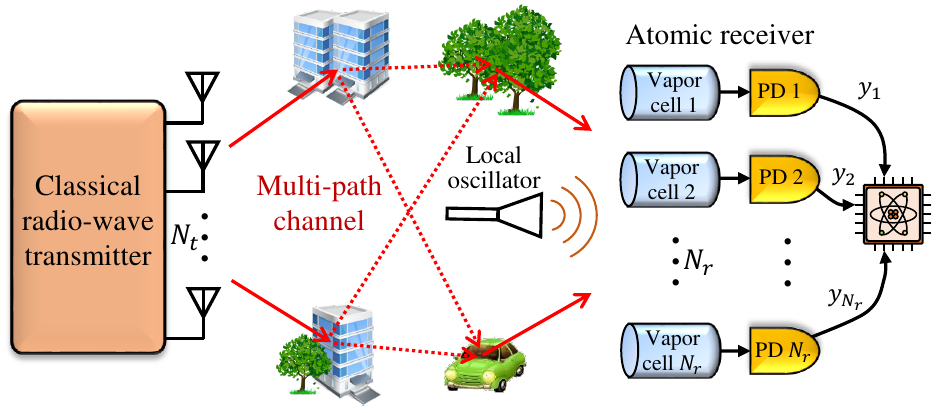}
\vspace*{-1em}
\caption{Rydberg atomic MIMO system.
}
\label{img:layout}
\vspace*{-1em}
\end{figure}

To begin with, we briefly review the system model of an atomic MIMO communication system. Readers can refer to our prior work \cite{AtomicMIMO_Cui2024} for a more comprehensive understanding of Rydberg atomic receivers. 
As illustrated in Fig.~\ref{img:layout}, the transmitter employs a conventional $N_t$-element antenna array for data transmission.  The time-domain radiated signal from the $n$-th antenna, $n\in\{1,2,\cdots,N_t\}$, is represented as $x_{n}(t) = A_n\cos(\omega t + \gamma_n)$. 
To realize phase demodulation at the receiver side,  the heterodyne 
measurement scheme, also known as the holographic-phase sensing measurement, is adopted~\cite{RydPhase_And2020}. This scheme introduces a LO as an auxiliary source to radiate predetermined reference signals to the atomic receiver. For ease of discussion, we assume that the reference signal,  $x_R(t)$, shares the same frequency with the transmitted signal, and thereby it is expressed as $x_R(t) = A_R \cos(\omega t + \gamma_R)$.

The atomic receiver deploys $N_r$ Rydberg atomic antennas for detecting the incident signals. Each antenna consists of a vapor cell filled with Alkali atoms. These atoms are excited to Rydberg states via two-photon transition process using a probe laser and a coupling laser\cite{QuanSense_Zhang2023}. The incident EM wave can strongly interact with the electric dipole moments of these Rydberg atoms, while the interaction results are read out by photodetectors (PDs) to perform signal detection~\cite{RydMag_Art2022}. Mathematically speaking, by adopting the multi-path wireless channel model, the EM wave $\mb{E}_m(t)$ impinging on the $m$-th atomic antenna, $m\in\{1,2,\cdots,N_r\}$, is written as
\begin{align}
    \mb{E}_m(t) &= \sum_{n=1}^{N_t}\sum_{\ell = 1}^L
    \boldsymbol{\epsilon}_{mnl}\rho_{mnl}A_{n}\cos(\omega t + \phi_{mnl} + \gamma_{n}) \notag\\
    & \quad\quad+ \boldsymbol{\epsilon}_{R,m}\rho_{R, m}A_{R}\cos(\omega t + \phi_{R,m} + \gamma_{R}).
\end{align}
% \begin{align}
%     \mb{E}_m(t) &= \sum_{n=1}^{N_t}\boldsymbol{\epsilon}_{m,n}\rho_{m,n}A_{n}\cos(\omega t + \phi_{m,n} + \gamma_{n}) + \boldsymbol{\epsilon}_{m,L}\rho_{m,L}A_{L}\cos(\omega t + \phi_{m,L} + \gamma_{L}).
% \end{align}
Here, $L$ denotes the number of transmitter-to-receiver channel paths. Variables $\boldsymbol{\epsilon}_{mnl}$, $\rho_{mnl}$, and $\phi_{mnl}$ represent the polarization direction, path loss, and phase shift of the radio wave propagating from the $n$-th transmit antenna to the $m$-th atomic antenna via the $\ell$-th path. Variables $\boldsymbol{\epsilon}_{R,m}$, $\rho_{R,m}$, and $\phi_{R,m}$ refer to the channel parameters pertaining to the reference signal. Since the LO is typically deployed nearby the receiver, we directly employ the line-of-sight path to model the LO-to-receiver channel. 
% and $\boldsymbol{\epsilon}_{m,L}$ represent  $\rho_{m,n}$ and $\rho_{m,L}$ the path losses, and $\phi_{m,n}$ and $\phi_{m,L}$ the phase shifts of the channels from the $n$-th transmit antenna and the LO to the 
% The incident EM wave can interact with the Rydberg atoms and trigger their electron transitions between different energy levels.

Suppose the signal frequency of the incident EM wave, $\omega$, is tuned to resonant with the transition frequency, $\omega_{ge}$, between two specific energy levels of Rydberg atoms, denoted as $\ket{g}$ and $\ket{e}$. Then, 
governed by the Lindblad master equation (Schrödinger equation in open quantum systems)~\cite{QuantumOptics_Fox2006}, the electrons of Rydberg atoms will oscillate between these two energy levels, a phenomenon referred to as \emph{electron transition}~\cite{AtomicMIMO_Cui2024}. 
The transition intensity is determined by a physical parameter called \emph{Rabi frequency}~\cite{QuantumOptics_Fox2006}: 
\begin{align}
    \Omega_m = \frac{|\boldsymbol{\mu}^T_{\rm RF}\mb{E}_{{\rm RF}, m}|}{\hbar}.
\end{align}
where $\boldsymbol{\mu}_{\rm RF}$ represents the electric dipole moment of Rydberg atoms, and $\hbar$ denotes the reduced Planck constant. Applying the rotating wave approximation~\cite{QuantumOptics_Fox2006}, our prior work \cite{AtomicMIMO_Cui2024} reveals that the vector $\mb{E}_{{\rm RF}, m}$ can be equivalently expressed as the down-converted complex form of $\mb{E}_{m}(t)$, i.e., 
% $\mb{E}_{{\rm RF}, m} = \sum_{n=1}^{N_t}\boldsymbol{\epsilon}_{m,n}\rho_{m,n}A_{n}e^{j(\phi_{m,n} + \gamma_{n})}+ \boldsymbol{\epsilon}_{m,L}\rho_{m,L}A_{L}e^{j(\phi_{m,L} + \gamma_{L})}$,
\begin{align}
    \mb{E}_{{\rm RF}, m} &= \sum_{n=1}^{N_t}\sum_{\ell = 1}^L\boldsymbol{\epsilon}_{mnl}\rho_{mnl}A_{n}e^{j(\phi_{mnl} + \gamma_{n})} \notag\\
    & \quad\quad\quad+\boldsymbol{\epsilon}_{R,m}\rho_{R,m}A_{R}e^{j(\phi_{R,m} + \gamma_{R})},
\end{align}
where the signal frequency, $\omega$, is physically cancelled by its resonant transition frequency, $\omega_{ge} = \omega$. 
Given that the signals, $\{\mb{E}_{{\rm RF}, m}\}$, are embedded into Rabi frequencies, $\{\Omega_m\}$, efficient signal detection hinges on measuring $\{\Omega_m\}$ accurately. 
This is popularly accomplished by a quantum coherence process called 
\emph{electromagnetically induced transparency (EIT)} spectroscopy~\cite{RydMag_Fancher2021}. 
Briefly speaking, the electron transition process changes the rate at which Rydberg atoms absorb photons from an incoming laser~\cite{RydMag_Fancher2021}. Thereby, by steering a probe laser to penetrate through the vapor cell and observing the dynamic change of its spectrum using a PD, the Rabi frequency can be read out. We denote the measurement of $[\Omega_1, \cdots, \Omega_{N_r}]^T$ as a vector, $\mb{y} = [y_1, \cdots, y_{N_r}]^T$, followed by
\begin{align}\label{eq:y}
    \mb{y} = |\mb{H}\mb{x} + \mb{r}|.
\end{align}
Here, $\mb{H}\in\mathbb{C}^{N_r\times N_t}$ is the effective channel matrix, whose $(m,n)$-th entry is $\mb{H}_{m,n} = \sum_{\ell = 1}^L 
\frac{1}{\hbar}\boldsymbol{\mu}_{\rm RF}^T\boldsymbol{\epsilon}_{mnl}\rho_{mnl}e^{j\phi_{mnl}}$, the vector $\mb{x} = [x_1, x_2, \cdots, x_{N_t}]^T$ with $x_n = A_n e^{j\gamma_n}$ denotes the transmitted signal in the complex form, and $\mb{r} = [r_1, r_2, \cdots, r_{N_r}]^T$ with $r_m = \frac{1}{\hbar}\boldsymbol{\mu}_{\rm RF}^T\boldsymbol{\epsilon}_{R, m}\rho_{R, m}A_{R}e^{j(\phi_{R,m} + \gamma_{R})}$, denotes the received reference signal. 

Finally, by invoking the law of large numbers, the measurement noise can be modelled as the complex Gaussian distribution, $\mb{w}\sim \mathcal{CN}(0, \sigma^2 \mb{I}_{N_r})$.  As a result, we arrive at the following input-output relationship for atomic MIMO systems:
\begin{align}\label{eq:model}
    \mb{y} = |\mb{H}\mb{x} + \mb{r} + \mb{w}|, 
\end{align}
which is a non-linear magnitude-detection model. 
% In our previous work, we proposed a signal detection algorithm to recover QAM symbols, $\mb{x}$, from this non-linear model \eqref{eq:model}. 
To analyze its achievable capacity, $C$, this article focuses on the transmitter precoding design to maximize the mutual information between $\mb{y}$ and $\mb{x}$:
\begin{align}\label{eq:MIM}
C = \max_{p(\mb{x})} \mathcal{I}(\mb{y}; \mb{x}) = \mathcal{H}(\mb{y}) - \mathcal{H}(\mb{y}|\mb{x}),
\end{align}
where $\mathcal{H}(\cdot)$ denotes the entropy of its arguments and $\mathcal{I}(\cdot;\cdot)$ the mutual information of its arguments.

    % manipulating the distribution of $\mb{x}$ via transmitter precoding to maximize the
    
     % the circularly symmetric Gaussian noise.

\section{Channel Capacity With Atomic MIMO Receivers}\label{sec:3}
In this section, we introduce a \emph{strong-reference approximation} to linearize the non-linear transition model in \eqref{eq:model}. This allows the mutual information, $\mathcal{I}(\mb{y}; \mb{x})$, as well as the channel capacity, $C$, to be explicitly formulated. 
% We then analyze the DoF of atomic MIMO systems relying 
% on the modeled mutual information. 

% to achieve the capacity limit of an atomic receiver, an IQ-aware precoding architecture is proposed. 

% \subsection{Approximated Mutual Information}\label{sec:3.1}

{For conventional MIMO systems, both the amplitude and phase information of received signal are accessible, resulting in a linear transition model:  $\mb{y} = \mb{H}\mb{x} + \mb{w}$. It has been proven that the channel capacity, $C'$, of this linear MIMO model is achievable if and only if $\mb{x}$ follows a \emph{circularly-symmetric} complex Gaussian (CSCG) distribution\footnote{A Gaussian random vector $\mb{x}$ is CSCG if and only if $e^{j\theta}\mb{x}$ has the same distribution with $\mb{x}$ for any real scalar $\theta$.}, i.e.,
\begin{align}\label{eq:MI0}
    C' = \max_{\mb{x}\sim \mathcal{CN}(0, \mb{Q})}  \log \det\left(
 \mb{I}_{N_r} + \frac{1}{\sigma^2} \mb{H} \mb{Q} \mb{H}^H.
 \right),
\end{align}
where $\mb{Q} = \mathsf{E}(\mb{x}\mb{x}^H)$. For an atomic MIMO receiver, however, it is evident from \eqref{eq:model} that it exhibits a non-linear input-output relationship, given that Rabi frequencies, $\{\Omega_m\}$, only capture the magnitude of incident EM waves. This fact makes it intractable to express the mutual information, $\mathcal{I}(\mb{y}; \mb{x})$, analytically, let alone calculating the channel capacity, $C$. }

To overcome the above challenge, we notice that the LO (see Fig.~\ref{img:layout}) is typically placed close to the atomic receiver \cite{RydPhase_Simons2019, Rydphase_Holloway2019,RydNP_Jing2020,chen_instantaneous_2024}. As a result, the LO-to-receiver distance (e.g., tens of centimeters) is much shorter than the transmitter-to-receiver distance (e.g., tens of meters), resulting in a much stronger reference signal than the wireless signal and noise. 
% \emph{the reference signal can be much stronger than the wireless data and noise}. 
This admits a strong-reference approximation to linearize the transmission model in \eqref{eq:model} as follows.

\begin{lemma}\label{prop:sra}
    (Strong-reference approximation): 
\emph{If each entry of the reference signal is greatly stronger than the wireless signal and noise, i.e., 
$|r_m| \gg |\mb{H}_{m,:}\mb{x} + w_m|$,  we can approximate $\mb{y}$ in \eqref{eq:model} as
\begin{align}\label{eq:real}
 \boxed{\mb{y} - |\mb{r}| \approx \bar{\mb{y}}= \mathrm{Re}(\widetilde{\mb{H}}\mb{x}) + \bar{\mb{w}},}
\end{align}
where $\widetilde{\mb{H}}_{m,:} \overset{\Delta}{=} e^{-j\angle r_m}\mb{H}_{m,:}$ and $\bar{\mb{w}} \sim \mathcal{N}(0, \frac{\sigma^2}{2}\mb{I}_{N_r}) $. }
\end{lemma}
\begin{proof}
(See Appendix \ref{appendix:sra}). 
\end{proof}
Lemma \ref{prop:sra} indicates that, in the strong-reference-signal regime, the \emph{non-linear} magnitude detector in \eqref{eq:model} can be transformed into a \emph{linear} real-part detector in \eqref{eq:real} by subtracting the offset $|\mb{r}|$. 
More importantly, the complicated mutual information, $\mathcal{I}(\mb{y}; \mb{x}) = \mathcal{I}({\mb{y}} - |\mb{r}|; \mb{x})$, is converted to a tractable expression: $\mathcal{I}({\mb{y}} - |\mb{r}|; \mb{x}) \rightarrow \mathcal{I}(\bar{\mb{y}}; \mb{x})$.
    % \begin{align} \label{eq:MI2}
    %     \mathcal{I}({\mb{y}} - |\mb{r}|; \mb{x}) \rightarrow \mathcal{I}(\bar{\mb{y}}; \mb{x}). 
    % \end{align}
    To write $\mathcal{I}(\bar{\mb{y}}; \mb{x})$ analytically, we introduce the subscripts $I$ and $Q$ to denote the real (in-phase) and imaginary (quadrature) components of a complex number, e.g., $\widetilde{\mb{H}} = \widetilde{\mb{H}}_I + j \widetilde{\mb{H}}_Q$ and $\mb{x} = \mb{x}_I + j \mb{x}_Q$. Then, the received signal $\bar{\mb{y}}$ in \eqref{eq:real} is rewritten as 
 \begin{align}\label{eq:real2}
\bar{\mb{y}} = (\widetilde{\mb{H}}_I, -\widetilde{\mb{H}}_Q)
    \left(\begin{array}{c}
      \mb{x}_I \\
      \mb{x}_Q
    \end{array}\right) + \mb{w}_I = {\bar{\mb{H}}}{\bar{\mb{x}}} + \bar{\mb{w}},
\end{align}
%  \begin{align}\label{eq:real2}
% \widetilde{\mb{y}} = \underbrace{(\widetilde{\mb{H}}_I, -\widetilde{\mb{H}}_Q)}_{\bar{\mb{H}}}
%     \underbrace{\left(\begin{array}{c}
%       \mb{x}_I \\
%       \mb{x}_Q
%     \end{array}\right)}_{\bar{\mb{x}}} + \mb{w}_I = {\bar{\mb{H}}}{\bar{\mb{x}}} + \mb{w}_I,
% \end{align}
where $\bar{\mb{H}}\overset{\Delta}{=} (\widetilde{\mb{H}}_I, -\widetilde{\mb{H}}_Q) \in\mathbb{R}^{N_r\times 2N_t}$ and $\bar{\mb{x}} \overset{\Delta}{=} ({\mb{x}}_I^T, {\mb{x}}_Q^T)^T \in\mathbb{R}^{2N_t\times 1}$ represent the \emph{real-valued}  channel and radiated signal, respectively. Attributed to the linear behavior of real-part detectors, the maximal throughput is achieved when the radiated signal $\bar{\mb{x}}$ follows the Gaussian distribution $\bar{\mb{x}} \sim \mathcal{N}(0, \bar{\mb{Q}})$, with the covariance matrix $\bar{\mb{Q}} = \mathsf{E}(\bar{\mb{x}}\bar{\mb{x}}^T) \in \mathbb{R}^{2N_t \times 2N_t}$. Thereafter, the mutual information is expressed as 
% As  noise $\mb{w}_I$ are real-valued, the general expression of the approximated mutual information is 
\begin{align}\label{eq:MI}
 \mathcal{I}(\bar{\mb{y}}; \mb{x}) = \frac{1}{2} \log \det(
 \mb{I}_{N_r} + \frac{2}{\sigma^2} \bar{\mb{H}} \bar{\mb{Q}} \bar{\mb{H}}^T
 ).
\end{align}
Here, the discount factor $\frac{1}{2}$ arises because the variables $\bar{\mb{H}}$, $\bar{\mb{x}}$, and $\bar{\mb{w}}$ are all real valued. 
% \begin{remark} (Approximated mutual information): 
Moreover, considering the total power constraint imposed on the radiated signal, we have 
\begin{align}\label{eq:P}
    \mathsf{Tr}\{\mathsf{E}(\mb{x}\mb{x}^H)\} = \mathsf{Tr}\{\mathsf{E}(\bar{\mb{x}}\bar{\mb{x}}^T)\} = \mathsf{Tr}\{\bar{\mb{Q}}\} = P,
\end{align}
where $P$ is the transmission power. 
    % Applying \textbf{Proposition \ref{prop:sra}}, we can approximate the mutual information between $\mb{y}$ and $\mb{x}$ as 
% \end{remark}
In summary, the channel capacity, $C$, of atomic MIMO systems can be formulated as 
\begin{align}\label{eq:MI3}
    \boxed{C = \max_{\mathsf{Tr}\{\bar{\mb{Q}}\} = P, \:\bar{\mb{Q}}\succeq0} \frac{1}{2} \log \det\left(
 \mb{I}_{N_r} + \frac{2}{\sigma^2} \bar{\mb{H}} \bar{\mb{Q}} \bar{\mb{H}}^T
 \right).}
\end{align}
\begin{remark} \emph{
    By comparing the channel capacities of conventional MIMO in \eqref{eq:MI0} and atomic MIMO in \eqref{eq:MI3}, one can observe that the  channel matrix is converted from a complex form, $\mb{H}$, to a real form, $\bar{\mb{H}}$. 
Moreover, the latter capacity is directly related to the radiated signals $\mb{x}_I$ and $\mb{x}_Q$ via their joint distribution, $\bar{\mb{x}}\sim\mathcal{N}(0, \bar{\mb{Q}})$, rather than their complex form, $\mb{x}\sim\mathcal{CN}(0, {\mb{Q}})$. 
% These facts imply that the optimal distribution of $\mb{x}_I$ and $\mb{x}_Q$ that achieves the channel capacity of atomic MIMO can potentially produce a circularly-symmetric distributed variable $\mb{x}$
These observations suggest that achieving the channel capacity of atomic MIMO requires that $\mb{x}_I$ and $\mb{x}_Q$ (or equivalently $\mb{x} = \mb{x}_I + j\mb{x}_Q$) follow a general Gaussian distribution, instead of a CSCG distribution in \eqref{eq:MI0}. This difference leads to a distinguishing feature of IQ-aware precoding for atomic MIMO.
}
% highlights a fundamental difference between precoding for atomic MIMO and precoding for conventional MIMO.} 
\end{remark}

\section{IQ-Aware Digital Precoding}\label{sec:4}
% \subsection{IQ-Aware Digial Precoding}
% When channel state information is available at both receiver and transmitter ends, we can manipulate the covariance matrix $\bar{\mb{Q}}$ on purpose using precoding technique to achieve the instantaneous capacity in \eqref{eq:MI3} for each channel realization.  

In this section, we propose an IQ-aware digital precoding method to achieve the channel capacity, $C$, in \eqref{eq:MI3} of atomic MIMO systems.
% approach designed to optimize digital precoding for atomic MIMO receivers. 
% Its comparison with the traditional digital precoding method is presented for a comprehensive understanding. 
The asymptotic channel capacity at high SNRs is subsequently analyzed to unveil the system DoF.

\subsection{Precoding Design}\label{sec:4.1}

% \subsection{Precoding Method}
We define $\mb{s}_I \in \mathbb{R}^{N_s}$ and $\mb{s}_Q \in \mathbb{R}^{N_s}$ as the data symbols in the in-phase and quadrature channels, respectively, where $N_s$ is the number of complex data streams.  
The overall baseband data is denoted as $\bar{\mb{s}} = [\mb{s}_I^T,\mb{s}_Q^T]^T \in\mathbb{R}^{2N_s}$ in the real form and $\mb{s} = \mb{s}_I + j\mb{s}_Q \in\mathbb{C}^{N_s}$ in the complex form.  The symbols $\mb{s}_I$ and $\mb{s}_Q$ follow i.i.d Gaussian distributions $\mathcal{N}(0, \frac{1}{2}\mb{I}_{N_s} )$, resulting in the CSCG distribution of $\mb{s}$. The objective of  precoding is to design a linear mapping $f:\bar{\mb{s}} \mapsto \bar{\mb{x}}$ such that the covariance matrix $\bar{\mb{Q}}$ can maximize the mutual information in \eqref{eq:MI3}. 

\subsubsection{Sub-optimality of traditional digital precoding}
% We start with briefly reviewing the implementation of  classical digital precoding.  
{  To achieve the channel capacity of a MIMO system, the most intuitive method is to implement \emph{traditional complex-valued digital precoding} as it has been demonstrated to be optimal in traditional MIMO systems. 
% Unfortunately, we will show that this method is not feasible for an atomic MIMO receiver.
This, however, is sub-optimal in the case of an atomic MIMO system.
% because it can only produce circularly-symmetric distributed  signal. 
Specifically, the classical precoding uses a complex matrix, $\mb{F} = \mb{F}_I + j\mb{F}_Q \in \mathbb{C}^{N_t \times N_{s}}$, to jointly process the baseband symbols $\mb{s}_I$ and $\mb{s}_Q$ before the IQ modulator, as shown in Fig.~\ref{img:tx} (a). 
After merging the IQ components, the transmit signal in the complex form is thereby the product of $\mb{F}$ and $\mb{s}$: 
\begin{align}\label{eq:x1}
    \mb{x} = \mb{F}\mb{s} = \mb{F}_I\mb{s}_I - \mb{F}_Q\mb{s}_Q + j(\mb{F}_Q\mb{s}_I + \mb{F}_I\mb{s}_Q),
\end{align}
whose covariance is $\mb{Q} = \mathsf{E}(\mb{x}\mb{x}^H) = \mb{F}\mb{F}^H$.  
% To see the reliance of the mutual information of atomic MIMO,
% For subsequent exposition, it 
% \eqref{eq:MI}, to the precoding matrix $\mb{F}$, it is necessary to expand the complex signal $\mb{x}$ in \eqref{eq:x1} to the real form:
For subsequent exposition, it is useful to write $\mb{x}$ in \eqref{eq:x1} as the equivalent real form:
% , it is necessary to 
% {\color{black} Because classical RF antennas can access both the real and imaginary parts of received signal, the complex-valued precoding is able to achieve the instantaneous capacity of classical MIMO by assigning $\mb{F}$ with the principal singular vectors of the complex channel, $\mb{H}$. However, this method is infeasible for atomic MIMO systems. Specifically, to maximize the mutual information \eqref{eq:MI3}, it is necessary to expand the complex signal $\mb{x}$ to the real form:
\begin{align}\label{eq:fs1}
    \bar{\mb{x}} = \left(\begin{array}{c}
      \mb{x}_I \\
      \mb{x}_Q
    \end{array}\right) = \underbrace{\left(\begin{array}{cc}
      \mb{F}_I   & -\mb{F}_Q \\
      \mb{F}_Q   &  \mb{F}_I
    \end{array}\right)}_{\mb{F}_r}
    {\left(\begin{array}{c}
      \mb{s}_I \\
      \mb{s}_Q
    \end{array}\right)}, 
\end{align}
whose covariance matrix is $\bar{\mb{Q}} = \mathsf{E}(\bar{\mb{x}}\bar{\mb{x}}^T) = \frac{1}{2}\mb{F}_r\mb{F}_r^T$. 
% Based on the capacity formulations \eqref{eq:MI0} and \eqref{eq:MI3} as well as the precoding manipulations \eqref{eq:x1} and \eqref{eq:fs1}, we are able to show the optimality of traditional complex-valued digital precoding to conventional MIMO and its infeasibility to atomic MIMO. 
The optimality of the precoder structure of $\mb{F}$ for traditional MIMO and its sub-optimality for atomic MIMO are explained as follows. 

\begin{figure*}[t!]
\centering
\includegraphics[width=6.5in]{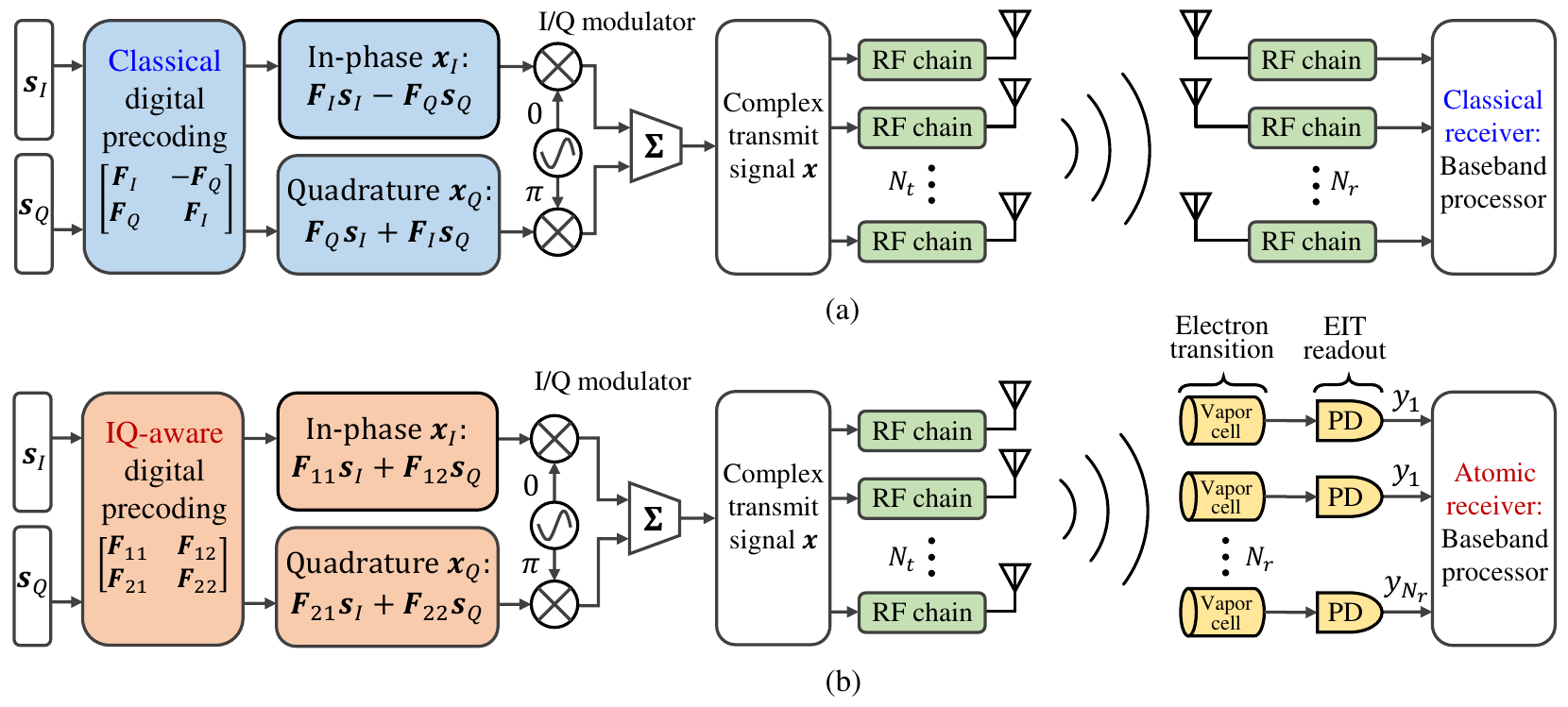}
\vspace*{-1em}
\caption{(a) Classical complex-valued digital precoding for classical receivers and (b) proposed IQ-aware digital precoding for atomic receivers. 
}
\label{img:tx}
\vspace*{-1em}
\end{figure*}

In a traditional MIMO system, the CSCG distribution is capacity achieving. Thus given CSCG distributed data symbols, $\mb{s}$, as input, a linear complex-valued precoder, $\mb{F}$, that retains the CSCG distribution can be optimal. Specifically, the optimal $\mb{F}$ is aligned with the principal singular-vectors of the channel $\mb{H}$. On the contrary, the CSCG precoded symbols cannot achieve the channel capacity of an atomic MIMO system, making the precoder $\mb{F}$ sub-optimal. To see the sub-optimality, one can observe from \eqref{eq:fs1} that the real-equivalence of $\mb{F}$, $\mb{F}_r$,  has a \emph{block-coupled constraint}, i.e., its diagonal block components are identical while its off-diagonal block components are opposite. Obviously, introducing this structural constraint reduces the channel capacity of atomic MIMO in \eqref{eq:MI3}. To elaborate, $\mb{F}_r$ under the constraint cannot always align with the principal singular-vectors of the real channel $\bar{\mb{H}}$ to achieve its capacity. For example, we consider a simple scenario: $N_r = N_t = 2$ and $N_s = 1$. The $2\times 2$ complex channel matrix is $\mb{H} = [2, j; j, 1]$ and the reference signal is $\mb{r} = [1; 1]$. In this context, the $2\times 4$ real channel matrix is $\bar{\mb{H}} = [\mb{H}_I, -\mb{H}_Q] = [2,0,0,-1; 0,1,-1,0]$. The principal singular vectors of $\bar{\mb{H}}$ are calculated as 
$\bar{\mb{v}}_1 = [-0.8944, 0, 0, 0.4472]^T$ and $\bar{\mb{v}}_2 = [0, -0.7071, 0.7071, 0]^T$. Evidently, the singular vectors $\bar{\mb{v}}_1$ and $\bar{\mb{v}}_2$ cannot form a precoding matrix satisfying the block-coupled constraint of $\mb{F}_r$. In this context, precoder $\mb{F}_r$ and its complex-equivalence $\mb{F}$ is not optimal.

\subsubsection{Proposed IQ-aware digital precoding}
The sub-optimality of conventional precoding motivates the proposed \emph{IQ-aware digital precoding} in Fig.~\ref{img:tx} (b), which relaxes the block-coupled constraint and thereby is capable of generating general-Gaussian symbols despite CSCG input. 
% To deal with , we propose . Its key idea is to relax the block-coupled constraint of the precoder in \eqref{eq:fs1}.
Specifically, compared with the joint manipulation of baseband IQ symbols using the complex precoder $\mb{F}$, the IQ-aware precoding employs four real-valued matrices, $\mb{F}_{k\ell}\in\mathbb{R}^{N_t\times N_s}, k\in\{1,2\},\ell\in\{1,2\}$, to independently process $\mb{s}_I$ and $\mb{s}_Q$.  By this means, the signals entering the IQ modulator become $\mb{F}_{11}\mb{s}_I + \mb{F}_{12}\mb{s}_Q$ and $\mb{F}_{21}\mb{s}_I + \mb{F}_{22}\mb{s}_Q$. The output signal in complex form is thereby
\begin{align} \label{eq:x2}
    \mb{x} = \mb{F}_{11}\mb{s}_I + \mb{F}_{12}\mb{s}_Q + j(\mb{F}_{21}\mb{s}_I + \mb{F}_{22}\mb{s}_Q).
\end{align}
% whose real form is thereby
One can verify that the equivalent real-form signal is 
\begin{align}\label{eq:fs2}
    \bar{\mb{x}} = \left(\begin{array}{c}
      \mb{x}_I \\
      \mb{x}_Q
    \end{array}\right) = \underbrace{\left(\begin{array}{cc}
      \mb{F}_{11}   &  \mb{F}_{12} \\
      \mb{F}_{21}   &  \mb{F}_{22}
    \end{array}\right)}_{\bar{\mb{F}}}
    {\left(\begin{array}{c}
      \mb{s}_I \\
      \mb{s}_Q
    \end{array}\right)}, 
\end{align}
whose covariance matrix is $\bar{\mb{Q}} = \frac{1}{2}\bar{\mb{F}}\bar{\mb{F}}^T$. 
The matrix $\bar{\mb{F}} \in \mathbb{R}^{2N_t \times 2N_s}$ represents the ``IQ-aware digital precoder". By comparing  \eqref{eq:x2} and \eqref{eq:fs2} with their counterparts in \eqref{eq:x1} and \eqref{eq:fs1}, two important features of IQ-aware digital precoding can be observed. First, the complex signal $\mb{x}$ in \eqref{eq:x2} is no longer the product of a complex precoder and a complex symbol as in \eqref{eq:x1}, making it possible for $\mb{x}$ to be general Gaussian distributed. Second, the equivalent real-valued precoder $\bar{\mb{F}}$ in \eqref{eq:fs2}, now is \emph{block-decoupled} with independent block components, $\{\mb{F}_{11}, \mb{F}_{12}, \mb{F}_{21}, \mb{F}_{22}\}$, which relaxes the block-coupled constraint of $\mb{F}_r$ in \eqref{eq:fs1}. 
These two features give the proposed IQ-aware digital precoding more flexibility in manipulating the transmitted signal $\mb{x}$, allowing for the utilization of SVD-based precoding technique to achieve the channel capacity. }

To be specific, we let $\bar{\mb{V}} \in \mathbb{R}^{2N_t\times 2N_t}$ be
the right-singular matrix of $\bar{\mb{H}}$, 
and $\{\bar{\lambda}_1,\cdots, \bar{\lambda}_{\min(N_r, 2N_t)}\}$ the corresponding singular values in a descending order.
The optimal IQ-aware digital precoder $\bar{\mb{F}}$ is represented as
\begin{align}\label{eq:Fopt}
    \bar{\mb{F}} =\left(\begin{array}{cc}
      \mb{F}_{11}   &  \mb{F}_{12} \\
      \mb{F}_{21}   &  \mb{F}_{22}
    \end{array}\right) = \bar{\mb{V}}_{:, 1:2N_s}\bar{\mb{P}}^\frac{1}{2}.
\end{align}
The power allocation matrix $\bar{\mb{P}} \overset{\Delta}{=} \diag\{\bar{p}_1,\cdots, \bar{p}_{2N_s}\}$ is determined by the water-filling principle:
% \begin{align}\label{eq:pk}
$\bar{p}_k = \left(\mu - \frac{\sigma^2}{\bar{\lambda}_k^2}\right)^{+}$, 
where the water level $\mu$ is adjusted to satisfy the power constraint $\mathsf{Tr}\{\bar{\mb{Q}}\} = \frac{1}{2}\mathsf{Tr}\{\bar{\mb{P}}\} = \frac{1}{2}\sum_{k= 1}^{2N_s} \bar{p}_k = P$. Consequently, by processing baseband data $\mb{s}_I$ and $\mb{s}_Q$ using $\{\mb{F}_{11}, \mb{F}_{12}, \mb{F}_{21}, \mb{F}_{22}\}$, the channel capacity in \eqref{eq:MI3} is achieved as:  
\begin{align}\label{eq:cap}
    C = \sum_{k = 1}^{2N_s} \frac{1}{2} \log\left(1 + \frac{\bar{\lambda}_k^2 \bar{p}_k}{\sigma^2}\right),
\end{align}
where the number of real-data streams, $2N_s$, should not exceed the rank of the real-channel matrix, denoted as $\min(N_r, 2N_t)$. 

\subsection{Channel Capacity Analysis}\label{sec:4.3}
% By substituting $\bar{\mathbf{F}}$ into equation \eqref{eq:MI3}, the channel capacity $C$ can be expressed as:
To gain more insights into atomic MIMO receivers, we analyze the asymptotic behavior of the channel capacity, $C$, at high SNRs
with the consideration of the maximal number of multiplexed data streams, i.e., $2N_s = \min(N_r, 2N_t)$.

\begin{lemma}\label{lemma:capacity} 
\emph{Define $\text{SNR} = \frac{P}{\sigma^2}$. When $\text{SNR}\rightarrow +\infty$ and $2N_s = \min(N_r, 2N_t)$, the channel capacity is asymptotically
        \begin{align}\label{eq:capacity}
        \boxed{C =  \min\left(\frac{N_r}{2}, N_t\right)\log\frac{\mathrm{SNR}}{ \min\left(\frac{N_r}{2}, N_t\right)} + \mathcal{O}(1).}
    \end{align}}
\end{lemma}
\begin{proof}
    Equation \eqref{eq:capacity} holds because $p_k \rightarrow \frac{P}{N_s}$ when $\text{SNR}\rightarrow +\infty$ and $\log(1 + x)\rightarrow \log(x)$ when $x\rightarrow +\infty$. 
\end{proof}
Lemma \ref{lemma:capacity} suggests that the capacity of an atomic MIMO system is improved by $\min(N_r/2, N_t)$ bit/s/Hz for every 3 dB increment in SNR, which corresponds to the DoFs of $\lim_{\text{SNR} \rightarrow \infty}\frac{C_e}{\log \text{SNR}} = \min\left(\frac{N_r}{2}, N_t\right)$~\cite{FundWC_Tse2015}. 
This result reflects the fact that half DoFs are lost at the receiver side due to the absence of imaginary-part information in \eqref{eq:real}, while the DoFs remain intact at the transmitter side. Thereafter, a discount factor $\frac{1}{2}$ is imposed on $N_r$, giving the DoFs of $\min(\frac{N_r}{2}, N_t)$. 
\begin{figure}[t!]
\centering
\includegraphics[width=3.5in]{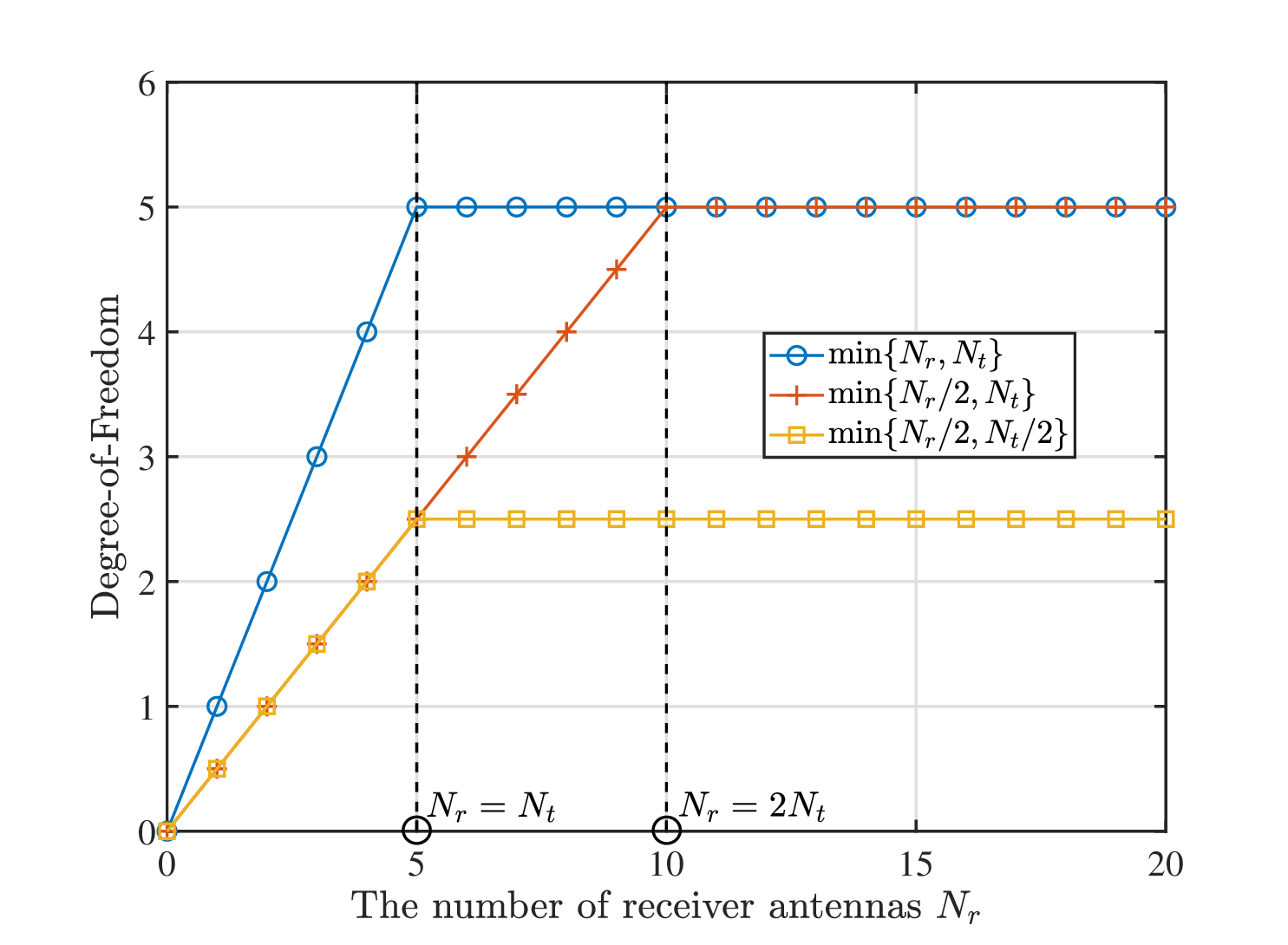}
\vspace*{-1em}
\caption{Comparison of DoFs between three cases: 1) $\min(N_r, N_t)$ for traditional MIMO, 2) $\min(\frac{N_r}{2}, N_t)$ for atomic MIMO, and 3)  
$\frac{1}{2}\min(N_r, N_t)$ for traditional MIMO with real symbols. The number of transimt antennas is set as $N_t = 5$.}
\label{img:dof}
\vspace*{-1em}
\end{figure}

For a better understanding, we compare three DoFs, $\min(N_r, N_t)$, $\min(\frac{N_r}{2}, N_t)$, and $\frac{1}{2}\min(N_r, N_t)$, as a function of $N_r$ in Fig.~\ref{img:dof}. 
They correspond to the following three cases, respectively: 1)
the traditional MIMO system, 2) the atomic MIMO system, 3) the In-phase transmission, where only \emph{real-valued symbols}, such as $\mb{s}_I$, are transmitted through a traditional $N_r\times N_t$ MIMO channel.
As $N_r$ increases, the DoFs $\min(\frac{N_r}{2}, N_t)$ gradually climb from $\frac{1}{2}\min(N_r, N_t)$ to $\min(N_r, N_t)$. Specifically, when $N_r < N_t$, we have $\min\left(\frac{N_r}{2}, N_t\right) = \frac{1}{2} \min\left({N_r}, N_t\right) = \frac{1}{2}N_r$. The DoFs of atomic MIMO systems are exactly half of that of conventional systems. In this context, it is sufficient to achieve the capacity of atomic receivers by transmitting $\frac{N_r}{2}$ complex-valued symbols, or equivalently $N_r$ real-valued symbols, such as $\mb{s}_I$. 
Moreover, as $N_r$ grows from $N_t$ to $2N_t$, the DoFs, $\min(N_r, N_t)$,  saturate while the gap between  $\min(N_r, N_t)$ and  $\min(N_r/2, N_t)$ narrows. 
% the gap of DoF becomes narrow. 
Particularly, when $N_r \ge 2N_t$, we have $\min(\frac{N_r}{2}, N_t) = \min(N_r, N_t) = N_t$. This relationship indicates that an atomic receiver does not suffer any DoF loss compared to its traditional counterpart if the number of atomic antennas is at least twice the number of transmit antennas. This result is intuitive because $N_r$ real-valued observations are sufficient to recover $N_t \le \frac{N_r}{2}$ complex-valued symbols. Consequently, under a same high SNR, the capacity gap between atomic and traditional MIMO systems tends to a constant independent of SNR. In summary, to fully exploit the multiplexing gain of an atomic MIMO system, it is necessary to have $N_r \ge 2 N_t$.

\section{IQ-Aware Hybrid Precoding}\label{sec:5}
For massive MIMO systems employing large-scale antenna arrays, hybrid precoding that incorporates digital and analog precoders is a more practical solution than fully digital precoding\cite{MIMO_Cui2022, OMP, HP_Yu2016, yu_alternating_2016}. 
Its essential idea is to cascade a low-dimensional digital processor (i.e., digital precoder) and a high-dimensional analog PS network (i.e., analog precoder), so as to reduce the number of costly radio-frequency (RF) chains. In this section, we extend the fully digital IQ-aware precoding in the preceding section to realize IQ-aware hybrid precoding. 
% In the preceding section, we have demonstrated that the capacity of atomic MIMO systems can be achieved using IQ-aware digital precoding. 
% % However, this is not the end of the story. 
% However, this solution presents additional challenges.
% Optimal IQ-aware digital precoding has to manipulate the IQ symbols independently before the IQ modulator. This requirement mandates its implementation on a digital baseband processor, because analog circuits, like phase shifters and power amplifiers, cannot distinguish the combined IQ symbols. Consequently, our IQ-aware digital precoding approach necessitates a \emph{fully digital precoding architecture}, where the number of RF chains is equal to the number of antennas. 
% However, employing a large number of RF chains results in prohibitive hardware costs and power consumption, especially in massive MIMO systems. 
% To mitigate these challenges, current 5G base stations (BSs) commonly adopt the \emph{hybrid analog and digital precoding architectures}, where precoding is executed by a low-dimensional digital processor followed by a high-dimensional analog PS network. 
% In this typical scenario, the IQ-aware precoder proposed in \eqref{eq:Fopt} is clearly inapplicable due to its reliance on digital processing.
Compared with traditional designs, a new challenge arises from the fact that the need of independent manipulation of IQ symbols cannot be implemented using analog circuits (i.e., a PS network). 

{  To tackle this challenge, the proposed IQ-aware hybrid precoding architecture deploys a low-dimensional IQ-aware digital precoder to assist the analog PS network. As illustrated in Fig.~\ref{img:HP},  both the FC and SC phase-shifter networks are considered. Compared with the IQ-aware fully digital precoder in \eqref{eq:fs2} where the entire precoding matrix is block-decoupled, the IQ-aware hybrid precoder is constructed by the product of a block-decoupled digital precoder and a block-coupled analog precoder. 
We will show that, despite the block-coupled constraint imposed on the analog part, the IQ-aware hybrid precoding is able to closely approximate the fully digital counterpart. }

% We will show that by meticularly optimizing 
% We will show that 
% These components are jointly optimized to approach the optimal IQ-aware digital precoder in \eqref{eq:Fopt}. 
% As illustrated in Fig.~\ref{img:HP}, our investigation focuses on two  hybrid precoding architectures: the FC and SC PS networks \cite{yu_alternating_2016}. 

% For modern 5G infrastructures, however, 
% This approach is not feasible for use in 5G base stations, where hybrid digital and analog precoding architectures are commonly adopted to significantly reduce the number of RF chains and hardware costs.

% Unfortunately, to reduce the hardware complexity and energy consumption 

% To support 
% This fact demands 
% The classic complex-valued precoder can be implemented either by digital signal processors prior to the IQ modulator or analog radio-frequency (RF) circuits after the IQ modulator, such as analog power amplifiers and phase shifters. However, it is notable that IQ-aware precoding requires independent processing of the IQ-channel symbols. \emph{This requirement implies that IQ-aware precoding must be deployed on the digital baseband prior to the IQ modulator}. 
% Based on this discussion, we can conclude that IQ-aware precoding is best suited for application on mobile edge devices with a limited number of antennas and affordable fully digital precoding circuits. Unfortunately, it is not feasible for use on a base station (BS) employing a large number of antennas supported by analog or hybrid precoding circuits. 

\subsection{IQ-Aware Hybrid Precoding with a FC-PS Network}\label{sec:5.1}
\subsubsection{Problem statement}
We first consider the IQ-aware hybrid precoding with a FC-PS network, where each RF chain is connected to all BS antennas via PSs. Denote the number of RF chains as $N_{\rm RF}$, satisfying $N_s \le N_{\rm RF} \ll N_t$. In  the baseband processor, an IQ-aware digital precoder $\bar{\mb{D}}$ of dimension $2N_{\rm RF} \times 2N_s$ is employed, with block components $\mb{D}_{k\ell}\in\mathbb{R}^{N_{\rm RF}\times N_s}, k\in\{1,2\},\ell\in\{1,2\}$, to precode the symbols $\mb{s}_I$ and $\mb{s}_Q$. The complex signal vector coming out of RF chains is thus $\mb{s}_{\rm rf} = \mb{D}_{11}\mb{s}_I + \mb{D}_{12}\mb{s}_Q + j(\mb{D}_{21}\mb{s}_I + \mb{D}_{22}\mb{s}_Q)$. At the RF front-end, the analog precoder is expressed by a complex matrix $\mb{A} = \mb{A}_I + j\mb{A}_Q\in\mathbb{C}^{N_t \times N_{\rm RF}}$. Given that PSs can only adjust the phases of signals, each element of $\mb{A}$ must satisfy the unit modulus constraint: $|\mb{A}_{m,n}|^2 = \mb{A}_{I,m,n}^2 + \mb{A}_{Q,m,n}^2 = 1,\forall m,n$.  Thereafter, the transmitted signal $\mb{x}$ is given as $\mb{x} = \mb{A}\mb{s}_{\rm rf}$, whose real form is 
% The radiated signal is denoted as 
\begin{align}\label{eq:xADs}
    \bar{\mb{x}} = \underbrace{\left(\begin{array}{cc}
      \mb{A}_I & -\mb{A}_Q\\
      \mb{A}_Q & \mb{A}_I
    \end{array}\right)}_{\bar{\mb{A}}} \underbrace{\left(\begin{array}{cc}
      \mb{D}_{11} & \mb{D}_{12}\\
      \mb{D}_{21} & \mb{D}_{22}
    \end{array}\right)}_{\bar{\mb{D}}} \underbrace{\left(\begin{array}{c}
      \mb{s}_I \\
      \mb{s}_Q 
    \end{array}\right)}_{\bar{\mb{s}}}.
\end{align} 
The covariance matrix of $\bar{\mb{x}}$ is $\bar{\mb{Q}} = \frac{1}{2}\bar{\mb{A}}\bar{\mb{D}}\bar{\mb{D}}^T\bar{\mb{A}}^T$. {  Comparing the precoding models \eqref{eq:fs2} and \eqref{eq:xADs}, one can observe that the analog part $\bar{\mb{A}}$ is subject to the block-coupled constraint, since analog PSs cannot distinguish the combined IQ signals. Fortunately, the hybrid precoder, $\bar{\mb{A}}\bar{\mb{D}}$, still has the ability to produce general Gaussian distributed signals $\bar{\mb{x}}$ because the digital part, $\bar{\mb{D}}$, is IQ-aware and has independent block components. This ability allows the accurate approximation of the IQ-aware fully digital precoder, $\bar{\mb{F}}$, by the hybrid counterpart, $\bar{\mb{A}}\bar{\mb{D}}$. 
% To realize this objective, 
% one can observe that the anlog precoder $\bar{\mb{A}}$ inherently possesses coupled diagonal and off-diagonal block components, which differentiates IQ-aware hybrid precoding from IQ-aware fully digital precoding. 
% To maximize the achievable rate, the analog precoder 
% the matrices $\bar{\mb{A}}$ and $\bar{\mb{D}}$ are jointly designed to approximate the IQ-aware fully digital precoder $\bar{\mb{F}}$:
To realize this objective, the joint optimization of the matrices $\bar{\mb{A}}$ and $\bar{\mb{D}}$ are formulated as
\begin{align}
   \min_{\bar{\mb{A}}, \bar{\mb{D}}}\:\: &\| \bar{\mb{F}} - \bar{\mb{A}}\bar{\mb{D}}\|_F^2, \label{eq:FCP1}\\
    \mathrm{s.t.}\:\:&\mb{A}_{I,m,n}^2 + \mb{A}_{Q,m,n}^2 =  1, \:\forall m,n, \tag{\ref{eq:FCP1}a} \label{eq:FCP1a} \\
    &\: \frac{1}{2} \mathsf{Tr}\{\bar{\mb{A}}\bar{\mb{D}}\bar{\mb{D}}^T\bar{\mb{A}}^T\} = P, \tag{\ref{eq:FCP1}b}\label{eq:FCP1b}
\end{align}
wherein \eqref{eq:FCP1b} is the total power constraint. In the subsequent discussions, we will devise an alternating minimization algorithm to solve $\bar{\mb{A}}\bar{\mb{D}}$ from problem \eqref{eq:FCP1} and use numerical simulations to show that the achieved rate of $\bar{\mb{A}}\bar{\mb{D}}$ can tightly approach the channel capacity in \eqref{eq:cap}.}
\begin{figure}[t!]
\centering
\includegraphics[width=3.5in]{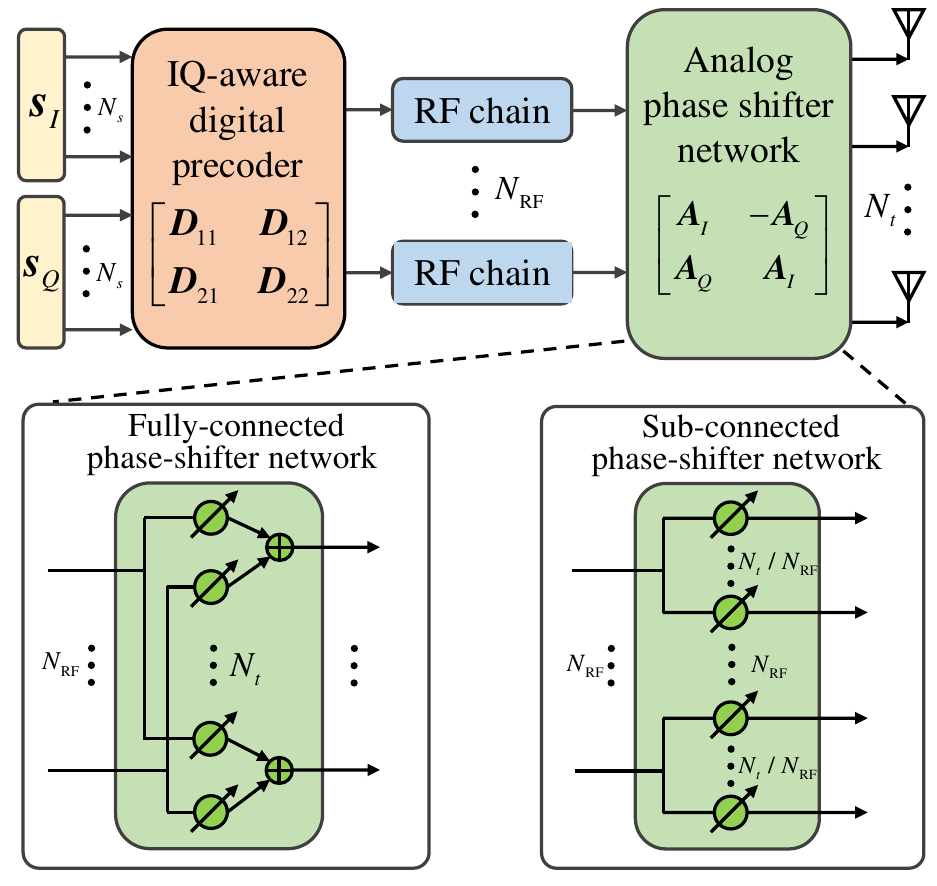}
\vspace*{-1em}
\caption{IQ-aware hybrid precoding with FC- and SC-PS networks.
}
\label{img:HP}
\vspace*{-1em}
\end{figure}

\subsubsection{Problem transformation}\label{sec:5.1.2}
Owing to the non-convex modulus constraint in \eqref{eq:FCP1a}, the block-coupled constraint of $\bar{\mb{A}}$, as well as the product of $\bar{\mb{A}}$ and $\bar{\mb{D}}$, directly solving problem \eqref{eq:FCP1} is challenging. To overcome these challenges, the first step is to seek an equivalent expression of \eqref{eq:FCP1} that decouples the analog precoder $\bar{\mb{A}}$ from its multiplication with $\bar{\mb{D}}$, which can significantly simplify the precoding design. 

Note that the columns of the optimal precoder $\bar{\mb{F}}$ are mutually orthogonal to mitigate the interference between data streams. 
The designed hybrid precoder $\bar{\mb{A}}\bar{\mb{D}}$ is thus expected to maintain column orthogonality as well. 
In massive MIMO systems, $N_t\gg N_{RF}$, the relationship $\bar{\mb{A}}^T\bar{\mb{A}} \approx N_t\mb{I}_{2N_{\rm RF}}$ holds with high probability~\cite{OMP,HP_Yu2016}. 
This is because the diagonal elements of $\bar{\mb{A}}^T\bar{\mb{A}}$ are exactly $N_t$ while the off-diagonal elements can be approximated as the summation of $2N_t$ independent terms, which is much less than $N_t$ with high probability for $N_t \gg 1$. These properties imply that the columns of the optimal $\bar{\mb{D}}$ are approximately orthogonal as well. Motivated by this observation, we introduce a column-orthogonal constraint on $\bar{\mb{D}}$ and thereby replace it by $\gamma\bar{\mb{D}}_u$, where $\bar{\mb{D}}^T_u\bar{\mb{D}}_u = \mb{I}_{2N_s}$. The scalar $\gamma$ is introduced to meet the total power constraint \eqref{eq:FCP1b}, whose value is approximately $\gamma \approx \sqrt{\frac{2P}{N_t}}$ by considering $\bar{\mb{A}}^T\bar{\mb{A}} \approx N_t\mb{I}_{2N_{\rm RF}}$. 
Consequently, the hybrid precoding design is recast as 
% \begin{align}
% \end{align}
\begin{align}
   \min_{\bar{\mb{A}}, \bar{\mb{D}}_u}\:\: &\| \bar{\mb{F}} - \gamma \bar{\mb{A}}\bar{\mb{D}}_u\|_F^2, \label{eq:FCP1.1}\\
   \mathrm{s.t.}\:\:&\mb{A}_{I,m,n}^2 + \mb{A}_{Q,m,n}^2 =  1, \:\forall m,n, \tag{\ref{eq:FCP1.1}a} \label{eq:FCP1.1a} \\
    &\: \bar{\mb{D}}^T_u\bar{\mb{D}}_u = \mb{I}_{2N_{s}}.  \tag{\ref{eq:FCP1.1}b}\label{eq:FCP1.1b}
\end{align}
% $$. 

Additionally, it is notable that although the columns of $\bar{\mb{D}}_u \in \mathbb{R}^{2 N_{\rm RF} \times 2N_s}$ are orthogonal, its rows are typically nonorthogonal given that  $N_{\rm RF}\ge N_s$, making it intractable to decouple the product of $\bar{\mb{A}}$ and $\bar{\mb{D}}_u$. To address this, we introduce two auxiliary matrices $\bar{\mb{D}}_c\in\mathbb{R}^{2N_{\rm RF} \times (2N_{\rm RF} - 2N_s)}$ and $\bar{\mb{F}}_c\in\mathbb{R}^{2N_{\rm t} \times (2N_{\rm RF} - 2N_s)}$, which complement $\bar{\mb{D}}_u$ to form a square matrix $\widetilde{\mb{D}} = (\bar{\mb{D}}_u, \bar{\mb{D}}_c) \in \mathbb{R}^{2N_{\rm RF}\times 2N_{\rm RF}}$ and complement $\bar{\mb{F}}$ to form $\widetilde{\mb{F}} = (\bar{\mb{F}}, \bar{\mb{F}}_c) \in \mathbb{R}^{2N_{\rm t}\times 2N_{\rm RF}}$, respectively. Then, the following lemma provides an equivalent expression for the objective function.
\begin{lemma} \label{lemma1}
    % The optimal solution of $\min_{\bar{\mb{A}}, \:\bar{\mb{D}}_u^T\bar{\mb{D}}_u = \mb{I}}\| \bar{\mb{F}} - \gamma \bar{\mb{A}}\bar{\mb{D}}_u\|_F^2$ is identical to 
    Given matrices $\bar{\mb{F}} \in \mathbb{R}^{2N_t \times 2N_{s}}$, $\bar{\mb{A}}\in\mathbb{R}^{2N_{t}\times2N_{\rm RF}}$, and $\bar{\mb{D}}_u \in \mathbb{R}^{2N_{\rm RF} \times 2N_{s}}$ satisfying  $\bar{\mb{D}}_u^T\bar{\mb{D}}_u = \mb{I}_{2N_s}$ and $N_{RF}\ge N_{s}$, the following optimization problems are identical:
    \begin{align}\label{eq:UOP}
     \min_{\bar{\mb{A}}, \bar{\mb{D}}_u}\| \bar{\mb{F}} - \gamma \bar{\mb{A}}\bar{\mb{D}}_u\|_F^2 = \min_{\bar{\mb{A}}, \widetilde{\mb{D}}, \bar{\mb{F}}_c}\| \widetilde{\mb{F}} - \gamma \bar{\mb{A}}\widetilde{\mb{D}}\|_F^2,
    \end{align}
    where $\widetilde{\mb{F}} = (\bar{\mb{F}}, \bar{\mb{F}}_c) \in \mathbb{R}^{2N_{\rm t}\times 2N_{\rm RF}}$, $\widetilde{\mb{D}} = (\bar{\mb{D}}_u, \bar{\mb{D}}_c) \in \mathbb{R}^{2N_{\rm RF}\times 2N_{\rm RF}}$, and $\widetilde{\mb{D}}^T\widetilde{\mb{D}} = \widetilde{\mb{D}}\widetilde{\mb{D}}^T = \mb{I}_{2N_{\rm RF}}$. 
\end{lemma}
\begin{proof}
    (See Appendix C). 
\end{proof}
Utilizing the auxiliary matrices, Lemma~\ref{lemma1} converts the optimization of the column-orthogonal matrix $\bar{\mb{D}}_u$ into optimizing the unitary matrix $\widetilde{\mb{D}}$. Since a unitary transformation preserves the power of its input, we have
\begin{align}\label{eq:unitary}
     \| \widetilde{\mb{F}} - \gamma \bar{\mb{A}}\widetilde{\mb{D}}\|_F^2 = \| \widetilde{\mb{F}}\widetilde{\mb{D}}^T - \gamma \bar{\mb{A}}\|_F^2.
    \end{align}
Equation \eqref{eq:unitary} gets the analog precoder $\bar{\mb{A}}$ rid of the product form with the digital precoder $\widetilde{\mb{D}}$. This simplification significantly facilitates our precoding algorithm design.     
    % In addition, it is notable that the normalization scalar $\gamma$ does not affect the shape of the precoder $\bar{\mb{A}}\bar{\mb{D}}_u$. Thereby, 
    % during the precoding design process, $\gamma$ can be set to $\gamma = 1$ to reduce the computational complexity. After obtaining the designed precoders $\bar{\mb{A}}$ and $\bar{\mb{D}}_u$, we can adjust $\gamma$ to $\sqrt{\frac{2P}{\mathsf{Tr}{\bar{\mb{A}}\bar{\mb{D}}_u\bar{\mb{D}}^T_u\bar{\mb{A}}^T}}}$ to meet the total power constraint $\frac{\gamma^2}{2} \mathsf{Tr}{\bar{\mb{A}}\bar{\mb{D}}_u\bar{\mb{D}}^T_u\bar{\mb{A}}^T} \le P$.  
By applying Lemma~\ref{lemma1} and \eqref{eq:unitary}, problem \eqref{eq:FCP1.1} is transformed into
\begin{align}
   \min_{\bar{\mb{A}},\widetilde{{\mb{D}}},\bar{\mb{F}}_c}
   \:\: &\| \widetilde{\mb{F}}\widetilde{\mb{D}}^T -  \gamma\bar{\mb{A}}\|_F^2 \label{eq:FCP2}\\
    \mathrm{s.t.}\:\:&\mb{A}_{I,m,n}^2 + \mb{A}_{Q,m,n}^2 =  1, \:\forall m,n, \tag{\ref{eq:FCP2}a} \label{eq:FCP2a} \\
    &\: \widetilde{\mb{D}}^T\widetilde{\mb{D}} = \widetilde{\mb{D}}\widetilde{\mb{D}}^T = \mb{I}_{2N_{\rm RF}}. \tag{\ref{eq:FCP2}b}\label{eq:FCP2b}
\end{align}
To efficiently solve this problem, we exploit the alternating minimization method. Specifically, the three variables $\bar{\mb{A}}$, $\widetilde{{\mb{D}}}$, and $\bar{\mb{F}}_c$ are updated alternatively while keeping the others fixed until convergence. 
If the updated solutions in each iteration are all optimal, the alternating minimization method is guaranteed to converge to a local optimum.

\subsubsection{Algorithm design} 

\begin{algorithm}[tb]
	\caption{$\!\!$: IQ-aware hybrid precoding with a FC-PS network}
 \label{alg1}
	\begin{algorithmic}[1]
		\REQUIRE ~ 
		The optimal IQ-aware fully digital precoder $\bar{\mb{F}}$.
		\STATE Initialize $\bar{\mb{F}}_c = \mb{0}$ and $\widetilde{\mb{I}}_c = \mb{I}_{2N_{\rm RF}}$,  and randomly initialize ($\mb{A}_I$, $\mb{A}_Q$); \\
  \STATE Calculate $\gamma = \sqrt{\frac{2P}{N_t}}$;
  \WHILE{not convergence}
    \STATE Update $\mb{A}_I$ and $\mb{A}_Q$ by \eqref{eq:AIQ}; \\
    % \STATE Update $\mb{P}_I$ and $\mb{P}_Q$ by \eqref{eq:popt1} and \eqref{eq:popt2}; \\
    \STATE Compute the SVD, $\widetilde{\mb{F}}^T\bar{\mb{A}} =\mb{V}_1\boldsymbol{\Sigma} \mb{V}_2^T$, and update $\widetilde{\mb{D}}$ by $\mb{V}_2\mb{V}_1^T$; \\
    \STATE Update $\bar{\mb{D}}_u = \widetilde{\mb{D}}_{:, 1:2N_s}$;
    \STATE Update $\bar{\mb{F}}_c = \gamma\bar{\mb{A}}\widetilde{\mb{D}}_{:, 1+2N_s:2N_{\rm RF}}$;
  \ENDWHILE
  \STATE $\bar{\mb{D}} = \sqrt{\frac{2P}{\mathsf{Tr}\{\bar{\mb{A}}\bar{\mb{D}}_u\bar{\mb{D}}^T_u\bar{\mb{A}}^T\}}} \bar{\mb{D}}_u$;
  \ENSURE ~ Optimized precoder $\bar{\mb{D}}$, $\mb{A}_I$, and $\mb{A}_Q$. \\
	\end{algorithmic}
\end{algorithm}

The step-by-step procedures of the proposed IQ-aware hybrid precoding algorithm for a FC-PS network are summarized in Algorithm~\ref{alg1}, with detailed explanations provided below.

\emph{\textbf{i) Fix ($\widetilde{{\mb{D}}}$, $\bar{{\mb{F}}}_c$) and optimize $\bar{\mb{A}}$}}: We start by deriving the optimal solution for $\bar{\mb{A}}$ after fixing matrices $\widetilde{{\mb{D}}}$ and $\bar{{\mb{F}}}_c$. Denote the block matrix components of $\widetilde{\mb{F}}\widetilde{\mb{D}}^T$ as $\mb{Z}_{k\ell}\in\mathbb{R}^{N_t\times N_{\rm RF}}$, $k\in\{1,2\},\:\ell\in\{1,2\}$, such that 
\begin{align} \label{eq:Z}
   \widetilde{\mb{F}}\widetilde{\mb{D}}^T  =
    {\left(\begin{array}{cc}
      \mb{Z}_{11} & \mb{Z}_{12}\\
      \mb{Z}_{21} & \mb{Z}_{22}
    \end{array}\right)}.
\end{align}
By substituting \eqref{eq:Z} into \eqref{eq:FCP2}, the objective function can be rewritten as 
\begin{align}\label{eq:obj_phase}
   &\| \widetilde{\mb{F}}\widetilde{\mb{D}}^T -  \gamma\bar{\mb{A}}\|_F^2 = \left\| \left(\begin{array}{cc}
      \mb{Z}_{11} - \gamma\mb{A}_I& \mb{Z}_{12} + \gamma\mb{A}_Q\\
      \mb{Z}_{21} -\gamma\mb{A}_Q& \mb{Z}_{22} -  \gamma\mb{A}_I
    \end{array}\right)\right\|_F^2 \notag\\
    &= \|\mb{Z}_I - \gamma\mb{A}_I\|_F^2 + \|\mb{Z}_Q -\gamma \mb{A}_Q\|_F^2 + C_1 \\
&\propto \sum_{m = 1}^{N_t}\sum_{n = 1}^{N_{\rm RF}} (\mb{Z}_{I,m,n} - \gamma\mb{A}_{I,m,n})^2 + (\mb{Z}_{Q,m,n} - \gamma\mb{A}_{Q,m,n})^2,\notag
\end{align}
wherein $\mb{Z}_I \overset{\Delta}{=} \mb{Z}_{11} + \mb{Z}_{22}$, $\mb{Z}_Q \overset{\Delta}{=} \mb{Z}_{21} - \mb{Z}_{12}$, and 
% $C_1 = \|\widetilde{\mb{F}}\widetilde{\mb{D}}^T\|_F^2 - \|\mb{Z}_I\|_F^2 - \|\mb{Z}_Q\|_F^2$. 
% $C_1 = \gamma^2 \|\mb{A}_I\|_F^2 + \gamma^2 \|\mb{A}_Q\|_F^2 - 2\mathsf{Tr}(\mb{Z}_{11}^T\mb{Z}_{22}) + 2\mathsf{Tr}(\mb{Z}_{21}^T\mb{Z}_{12}) = \gamma^2N_tN_{\rm RF}- 2\mathsf{Tr}(\mb{Z}_{11}^T\mb{Z}_{22}) + 2\mathsf{Tr}(\mb{Z}_{21}^T\mb{Z}_{12})$. 
$C_1 = \gamma^2 \|\mb{A}_I\|_F^2 + \gamma^2 \|\mb{A}_Q\|_F^2 - 2\mathsf{Tr}(\mb{Z}_{11}^T\mb{Z}_{22}) + 2\mathsf{Tr}(\mb{Z}_{21}^T\mb{Z}_{12}) = \gamma^2N_tN_{\rm RF}- 2\mathsf{Tr}(\mb{Z}_{11}^T\mb{Z}_{22}) + 2\mathsf{Tr}(\mb{Z}_{21}^T\mb{Z}_{12})$. 
Equations \eqref{eq:obj_phase} and \eqref{eq:FCP2a} indicate that the joint optimization of analog precoder, $\mb{A}$, is equivalent to  independently optimizing its elements. Consider the PS connecting the $m$-th antenna and the $n$-th RF chain. The optimal precoder coefficient is obtained by solving 
% Since the analog precoder $\bar{\mb{A}}$ gets rid of the product form with $\widetilde{\mb{D}}$ in \eqref{eq:}
\begin{align}
   \min_{\mb{A}_{m,n}}
   \:\: &(\frac{1}{\gamma}\mb{Z}_{I,m,n} - \mb{A}_{I,m,n})^2 + (\frac{1}{\gamma}\mb{Z}_{Q,m,n} - \mb{A}_{Q,m,n})^2 \label{eq:FCP2.1}\\
   \mathrm{s.t.}\:\:&\mb{A}_{I,m,n}^2 + \mb{A}_{Q,m,n}^2 =  1. \tag{\ref{eq:FCP2.1}a} \label{eq:FCP2.1a}
\end{align}
Problem \eqref{eq:FCP2.1} can be interpreted as finding a point $(\mb{A}_{I,m,n}, \mb{A}_{Q,m,n})$ on the unit circle that is closest to $(\frac{1}{\gamma}\mb{Z}_{I,m,n}, \frac{1}{\gamma}\mb{Z}_{Q,m,n})$, whose optimal solution is 
% clearly \eqref{eq:AIQ}.
% By applying some matrix techniques, we can prove that the coefficients of the optimal $\bar{\mb{A}}$ are analytically expressed as 
\begin{align}\label{eq:AIQ}
    \left\{ 
    \begin{array}{c}
         \mb{A}^\star_{I,m,n} = \frac{\mb{Z}_{I,m,n}}{\sqrt{\mb{Z}_{I,m,n}^2 + \mb{Z}_{Q,m,n}^2}}  \\
         \mb{A}_{Q,m,n}^\star = \frac{\mb{Z}_{Q,m,n}}{\sqrt{\mb{Z}_{I,m,n}^2 + \mb{Z}_{Q,m,n}^2}}
    \end{array}
    \right.
    .
\end{align}
% where $\mb{Z}_I \overset{\Delta}{=} \mb{Z}_{11} + \mb{Z}_{22}$ and $\mb{Z}_Q \overset{\Delta}{=}\mb{Z}_{21} - \mb{Z}_{12}$. The detailed proof is provided in Appendix D.
This completes the update of $\mb{A}$ in Algorithm 1.

\textbf{\emph{ii) Fix ($\bar{{\mb{A}}}$, $\bar{{\mb{F}}}_c$) and optimize $\widetilde{\mb{D}}$}}: With fixed matrices $\bar{{\mb{A}}}$ and $\bar{{\mb{F}}}_c$,  the objective function in \eqref{eq:FCP2} can be rewritten as 
\begin{align}
    &\|\widetilde{\mb{F}}\widetilde{\mb{D}}^T - \gamma \bar{\mb{A}}\|_F^2 \notag\\ &= \mathsf{Tr}\{\widetilde{\mb{F}}^T\widetilde{\mb{F}}\widetilde{\mb{D}}^T\widetilde{\mb{D}}\} + \gamma^2\mathsf{Tr}\{\bar{\mb{A}}^T\bar{\mb{A}}\} 
    - 2\gamma\mathsf{Tr}\{\widetilde{\mb{D}} \widetilde{\mb{F}}^T\bar{\mb{A}}\} \notag \\
    &\overset{(a)}{=} \mathsf{Tr}\{\widetilde{\mb{F}}^T\widetilde{\mb{F}}\} + 2N_tN_{\rm RF}\gamma^2 - 2\gamma\mathsf{Tr}\{\widetilde{\mb{D}} \widetilde{\mb{F}}^T\bar{\mb{A}}\},
\end{align}
where (a) holds because $\widetilde{\mb{D}}^T\widetilde{\mb{D}} = \mb{I}_{2N_{\rm RF}}$ and $\mathsf{Tr}\{\bar{\mb{A}}^T\bar{\mb{A}}\} = 2N_tN_{\rm RF}$. 
Accordingly, the sub-problem for optimizing $\widetilde{\mb{D}}$ is equivalent to 
\begin{align}
    \max_{\widetilde{{\mb{D}}}^T\widetilde{\mb{D}} = \mb{I}_{2N_{\rm RF}}} \:\:&\mathsf{Tr}\{\widetilde{\mb{D}}\widetilde{\mb{F}}^T\bar{\mb{A}}\} \label{eq:OPP} 
    % \mathrm{s.t.}\:\:& \widetilde{{\mb{D}}}^T\widetilde{\mb{D}} = \mb{I}_{2N_{\rm RF}}. \tag{\ref{eq:OPP}a} \label{eq:OPPa}
\end{align}
% \begin{align}
%     \max_{\widetilde{{\mb{D}}}} \:\:&\mathsf{Tr}\{\widetilde{\mb{D}}\widetilde{\mb{F}}^T\bar{\mb{A}}\} \label{eq:OPP} \\ 
%     \mathrm{s.t.}\:\:& \widetilde{{\mb{D}}}^T\widetilde{\mb{D}} = \mb{I}_{2N_{\rm RF}}. \tag{\ref{eq:OPP}a} \label{eq:OPPa}
% \end{align}
Problem \eqref{eq:OPP} is an orthogonal Procrustes Problem.  To solve it, we can compute the SVD: $\widetilde{\mb{F}}^T\bar{\mb{A}} = \mb{V}_1\boldsymbol{\Sigma} \mb{V}_2^T$. Then, the objective function in \eqref{eq:OPP} is upper bounded by 
\begin{align}
    &\mathsf{Tr}\{\widetilde{\mb{D}}\widetilde{\mb{F}}^T\bar{\mb{A}}\} = \mathsf{Tr}\{\widetilde{\mb{D}}\mb{V}_1\boldsymbol{\Sigma} \mb{V}_2^T\} = \mathsf{Tr}\{ \mb{V}_2^T\widetilde{\mb{D}}\mb{V}_1\boldsymbol{\Sigma}\} \notag \\& \quad\quad\quad= \sum_{n = 1}^{2N_{\rm RF}} (\mb{V}_2^T\widetilde{\mb{D}}\mb{V}_1)_{n,n} \boldsymbol{\Sigma}_{n,n}
    \overset{(b)}{\le}
    \sum_{n = 1}^{2N_{\rm RF}} \boldsymbol{\Sigma}_{n,n},
\end{align}
where (b) arises because  $\mb{V}_2^T\widetilde{\mb{D}}\mb{V}_1$ is an orthogonal matrix. The equality holds if and only if $\mb{V}_2^T\widetilde{\mb{D}}\mb{V}_1 = \mb{I}_{2N_{\rm RF}}$.
% Then, as proven in Appendix E, 
Henceforth, the optimal $\widetilde{\mb{D}}$ is obtained as 
\begin{align} \label{eq:Dopt}
    \widetilde{\mb{D}}^\star = \mb{V}_2\mb{V}_1^T,
\end{align}
which allows us to update $\bar{\mb{D}}_u^\star$ by $\widetilde{\mb{D}}^\star_{:, 1:2N_s}$ in step 6 of Algorithm 1.

\textbf{\emph{iii) Fix ($\widetilde{{\mb{D}}}$, $\bar{\mb{A}}$) and optimize $\bar{{\mb{F}}}_c$}}: We finally consider updating the auxiliary matrix $\bar{{\mb{F}}}_c$ given $\widetilde{{\mb{D}}}$ and $\bar{\mb{A}}$. To this end, we explicitly express the the objective function $\| \widetilde{\mb{F}} -  \bar{\mb{A}}\widetilde{\mb{D}}\|_F^2$ with respect to $\bar{{\mb{F}}}_c$  as 
\begin{align}
\| (\bar{\mb{F}}, \bar{\mb{F}}_c) -  \gamma \bar{\mb{A}}(\bar{\mb{D}}_u, \bar{\mb{D}}_c)\|_F^2 = \| \bar{\mb{F}}_c -  \gamma\bar{\mb{A}}\bar{\mb{D}}_c\|_F^2 + C_2, 
\end{align}
where $C_2 = \|\bar{\mb{F}} -  \gamma \bar{\mb{A}}\bar{\mb{D}}_u\|_F^2$ is irrelevant to $\bar{\mb{F}}_c$. 
It is evident that the optimal $\bar{\mb{F}}_c$ needs to cancel the least square error introduced by  $\gamma\bar{\mb{A}}\bar{\mb{D}}_c$:
\begin{align}
    \bar{\mb{F}}_c^\star = \gamma\bar{\mb{A}}\bar{\mb{D}}_c = \gamma \bar{\mb{A}}\widetilde{{\mb{D}}}_{:, 1+2N_s:2N_{\rm RF}}.
\end{align}
This completes the update of $\bar{\mb{F}}_c$ in step 7. 

\textbf{\emph{iv) Power normalization}}:
Recall that the previously calculated $\gamma$ is an approximated value. 
After the alternating minimization algorithm converges, 
% the optimized precoders might not perfectly satisfy the total power constraint. To address this issue, 
we normalize the digital precoder $\bar{\mb{D}}_u$ by $\sqrt{\frac{2P}{\mathsf{Tr}\{\bar{\mb{A}}\bar{\mb{D}}_u\bar{\mb{D}}^T_u\bar{\mb{A}}^T\}}}\bar{\mb{D}}_u$ to guarantee the total power constraint. This completes the design of IQ-aware hybrid precoding with a FC-PS network.

\subsection{IQ-Aware Hybrid Precoding with a SC-PS Network}\label{sec:5.2}
\subsubsection{Problem statement}
To show that our proposed IQ-aware hybrid precoding is feasible to general analog circuits, we now turn to the investigation of SC-PS networks. In this architecture, each RF is connected to a subset of the BS array with
$K = N_t/N_{\rm RF}$ antennas. The number of required PSs is thus $N_t$. With much reduced PSs, the SC architecture exhibits lower spectral efficiency than the FC one, but is more energy-efficient.  
% To simplify notation, we abuse the same expressions for the digital precoder $\bar{\mb{D}}$ and analog precoder $\bar{\mb{A}}$ as used in the FC-PS structure. 
The sub-connected feature states that the analog precoder $\mb{A} = \mb{A}_I + j\mb{A}_Q$ is now a block diagonal matrix: 
\begin{align}
    \mb{A} = {\rm blkdiag}\{\mb{p}_1, \mb{p}_2, \cdots, \mb{p}_{N_{\rm RF}}\},
\end{align}
where $\mb{p}_n = [{\rm exp}(j\theta_{n,1}), \cdots, {\rm exp}(j\theta_{n,K})]^T$ represent the PSs connected to the $n$-th RF chain. 
% The block diagonal feature allows us to prove that the columns of $\bar{\mb{A}}$ are orthogonal, i.e., $\bar{\mb{A}}^T\bar{\mb{A}} = K\mb{I}_{2N_{\rm RF}}$. 
Thanks to this block diagonal property, one can verify that $\mb{A}^H\mb{A} = K\mb{I}_{N_{\rm RF}}$ and $\bar{\mb{A}}^T\bar{\mb{A}} = K\mb{I}_{2N_{\rm RF}}$ with ease. These facts automatically releases $\bar{\mb{A}}$ from the total power constraint because $ \frac{1}{2} \mathsf{Tr}\{\bar{\mb{A}}\bar{\mb{D}}\bar{\mb{D}}^T\bar{\mb{A}}^T\} =  \frac{K}{2} \mathsf{Tr}\{\bar{\mb{D}}\bar{\mb{D}}^T\} = P$. Henceforth, the SC hybrid precoding design problem is formulated as  
\begin{align}
   \max_{\bar{\mb{A}}, \bar{\mb{D}}}\:\: &\|\bar{\mb{F}} - \bar{\mb{A}}\bar{\mb{D}}\|_F^2 \label{eq:SCP1}\\
    \mathrm{s.t.}\:\:&\mb{A} = {\rm blkdiag}\{\mb{p}_1, \mb{p}_2, \cdots, \mb{p}_{N_{\rm RF}}\}, \tag{\ref{eq:SCP1}a} \label{eq:SCP1a} \\
    &\mathsf{Tr}\{\bar{\mb{D}}\bar{\mb{D}}^T\} = {2P}/{K}.\tag{\ref{eq:SCP1}b}\label{eq:SCP1b}
\end{align}
% The structural similarity between problems \eqref{eq:FCP1} and \eqref{eq:SCP1} implies that we can utilize a similar alternating minimization approach to solve \eqref{eq:SCP1}, which are elaborated below.
% Problem \eqref{eq:SCP1} exhibits a structural similarity with problem \eqref{eq:FCP1}. 
Clearly, problem \eqref{eq:SCP1} differs from the FC hybrid precoding problem \eqref{eq:FCP1} by two aspects: 1) $\mb{A}$ is block diagonal and 2) $\mb{A}$ is independent of the total power constraint. 
% Notably, the total power constraint is now independent of the analog precoder $\bar{\mb{A}}$. 
These differences enable us to  alternatively optimize precoders $\bar{\mb{A}}$ and $\bar{\mb{D}}$ to solve \eqref{eq:SCP1}, without the need of performing problem transformation as in Section \ref{sec:5.1.2}.

\subsubsection{Algorithm design} The step-by-step procedures of our SC hybrid precoding algorithm is summarized in Algorithm~\ref{alg2}, while the details are explained below. 

% For clarity, we present the procedures of the proposed IQ-aware hybrid precoding algorithm with SC-PS network in Algorithm~\ref{alg2}, 
% As summarized in Algorithm~\ref{alg2}, our SC  alternatively updates the analog and digital precoders while fixing the other one.

% while the optimization details are elaborated below.

\begin{algorithm}[tb]
	\caption{$\!\!$: IQ-aware hybrid precoding with a SC-PS network}
 \label{alg2}
	\begin{algorithmic}[1]
		\REQUIRE ~ 
		The optimal IQ-aware fully digital precoder $\bar{\mb{F}}$.
		\STATE Randomly initialize $\mb{A}_I$, $\mb{A}_Q$, and $\bar{\mb{D}}$. \\
  \WHILE{not convergence}
    \STATE Update $\mb{A}_I$ and $\mb{A}_Q$ by \eqref{eq:SC_phase}; \\
    \STATE Update $\bar{\mb{D}}$ by \eqref{eq:Dopt2};
  \ENDWHILE
  % \STATE $\bar{\mb{D}} = \sqrt{\frac{2PN_{\rm RF}}{N_t\mathsf{Tr}{\bar{\mb{D}}\bar{\mb{D}}^T}}} \bar{\mb{D}}$;
  \ENSURE ~ Optimized precoder $\bar{\mb{D}}$, $\mb{A}_I$, and $\mb{A}_Q$. \\
	\end{algorithmic}
\end{algorithm}

\textbf{\emph{i) Fix $\bar{\mb{D}}$ and optimize $\bar{\mb{A}}$}}:
To find the optimal $\bar{\mb{A}}$,  we first 
substitute $\bar{\mb{A}}^T\bar{\mb{A}} = K\mb{I}_{2N_{\rm RF}}$ into \eqref{eq:SCP1} and rewrite the objective function as
\begin{align}
    \|\bar{\mb{F}} - \bar{\mb{A}}\bar{\mb{D}}\|_F^2 = C_3 - 2\mathsf{Tr}\{\bar{\mb{D}}\bar{\mb{F}}^T\bar{\mb{A}}\}, \label{eq:phase2}
\end{align} 
where $C_3 = \mathsf{Tr}\{\bar{\mb{F}}^T\bar{\mb{F}}\} +K\mathsf{Tr}\{\bar{\mb{D}}^T\bar{\mb{D}}\}$. Then, by defining $\mb{Y}_{k\ell}\in\mathbb{R}^{N_t\times N_{\rm RF}}$, $k\in\{1,2\},\:\ell\in\{1,2\}$ as the block matrix components of $\bar{\mb{F}}\bar{\mb{D}}^T$ such that
\begin{align}
   \bar{\mb{F}}\bar{\mb{D}}^T  =
    {\left(\begin{array}{cc}
      \mb{Y}_{11} & \mb{Y}_{12}\\
      \mb{Y}_{21} & \mb{Y}_{22}
    \end{array}\right)},
\end{align}
we can expand $\mathsf{Tr}\{\bar{\mb{D}}\bar{\mb{F}}^T\bar{\mb{A}}\}$ as  
\begin{align}\label{eq:phase4}
    &\mathsf{Tr}\left\{
    \left(
    \begin{array}{cc}
       \mb{Y}_{11}^T  &  \mb{Y}_{21}^T\\
       \mb{Y}_{12}^T  &  \mb{Y}_{22}^T
    \end{array}
    \right)\left(\begin{array}{cc}
      \mb{A}_I   & -\mb{A}_Q \\
      \mb{A}_Q   &  \mb{A}_I \\
    \end{array}\right)
    \right\} \notag \\
    & = \mathsf{Tr}\{
    (\mb{Y}_{11} + \mb{Y}_{22})^T\mb{A}_I + (\mb{Y}_{21} - \mb{Y}_{12})^T\mb{A}_Q
    \}  \notag \\
    & \overset{(a)}{=} \mathsf{Tr}\{\Re(\mb{Y}^H\mb{A})\} \overset{(b)}{\propto} -\frac{1}{2}\|\mb{Y} -  \mb{A}\|_F^2, 
\end{align}
where (a) arises due to the definition $\mb{Y}  \overset{\Delta}{=} (\mb{Y}_{11} + \mb{Y}_{22}) + j(\mb{Y}_{21} - \mb{Y}_{12})$, and (b) holds by adding the constant $-\frac{1}{2}(\|\mb{Y}\|_F^2 + \|\mb{A}\|_F^2)$ into $ \mathsf{Tr}\{\Re(\mb{Y}^H\mb{A})\}$. 
By further considering the block-matrix property of $\mb{A}$ and the unit modulus constraint of PSs, the optimal phases $\{\theta_{n,k}\}$ in precoder $\mb{A}$ are thereby
\begin{align}\label{eq:SC_phase}
    &\quad\quad\quad\theta_{n, k}^\star = \arg\left\{\mb{Y}_{ (n - 1)K + k,  n}\right\}, \notag \\
    &\quad\quad\quad\quad\quad\quad\quad\quad  1\le n\le N_{\rm RF},\: 1\le k \le K.
\end{align}

\textbf{\emph{ii) Fix $\bar{\mb{A}}$ and optimize $\bar{\mb{D}}$}}: For a fixed analog precoder, the optimization of $\bar{\mb{D}}$ is a quadratic programming problem with quadratic constraint. The optimal solution can be derived using the Karush-Kuhn-Tucker (KKT) conditions~\cite{ConvexProgramming} as 
\begin{align}\label{eq:Dopt2}
   \bar{\mb{D}}^\star &= \left(\bar{\mb{A}}^T\bar{\mb{A}} + \lambda^\star \mb{I}_{2N_{\rm RF}} \right)^{-1}\bar{\mb{A}}^T\bar{\mb{F}} = \frac{1}{K + \lambda^\star} \bar{\mb{A}}^T\bar{\mb{F}} \notag \\
   &\overset{(a)}{=}\sqrt{\frac{2P}{K\|\bar{\mb{A}}^T\bar{\mb{F}}\|_F^2}}\bar{\mb{A}}^T\bar{\mb{F}}.
\end{align}
wherein (a) holds by selecting a Lagrange multiplier $\lambda^\star$ to fulfill the total power constraint \eqref{eq:SCP1b}. This completes the design of IQ-aware hybrid precoding with a SC-PS network.

\subsection{Convergence and Complexity Analysis}\label{sec:5.3}
\subsubsection{Convergence analysis}\label{sec:5.3.1}
% It is of interest to show the convergence of the proposed algorithms. 
For both Algorithms 1 and 2, the updated solutions in each iteration can monotonically decrease the objective functions $\|\widetilde{\mb{F}}\widetilde{\mb{D}}^T - \bar{\mb{A}}\|_F^2$ and $\|\bar{\mb{F}} - \bar{\mb{A}}\bar{\mb{D}}\|_F^2$. Moreover, these two functions have a lower bound of zero. By invoking the monotone convergence theorem, the convergence of Algorithms 1 and 2 are guaranteed.

\subsubsection{Computational complexity analysis}\label{sec:5.3.2}
 % are largely attributed to alternative optimization procedures.
To evaluate the computational complexities of Algorithms~\ref{alg1} and \ref{alg2}, we access the number of scalar multiplications required by each algorithm.
Consider one iteration of Algorithm~\ref{alg1}. Updating $\mb{A}_I$ and $\mb{A}_Q$ 
necessitates the calculation of matrices $\mb{Z}_I$ and $\mb{Z}_Q$, which has a complexity of $\mathcal{O}(N_tN_{\rm RF}^2)$. 
Next, the optimization of $\widetilde{\mb{D}}_u$ and $\widetilde{\mb{F}}_c$ involves the calculation of matrix $\widetilde{\mb{F}}^T\widetilde{\mb{A}}$ and its SVD, resulting in a complexity on the order of $\mathcal{O}(N_t N_{\rm RF}^2 + N_{\rm RF}^3)$.
In summary, the overall complexity of Algorithm~\ref{alg1} is $\mathcal{O}(I_0(N_tN_{\rm RF}^2  +N_{\rm RF}^3))$, where $I_0$ is the number of iterations required for convergence. 
In addition, consider one iteration of Algorithm~\ref{alg2}.  
The complexity for updating $\mb{A}_I$ and $\mb{A}_Q$ primarily results from computing the expression of matrix $\mb{Y}$. This step has a complexity of $\mathcal{O}(N_t N_{\rm RF}N_s)$. 
Subsequently, the complexity of updating $\bar{\mb{D}}$ by \eqref{eq:Dopt} is on the order of $\mathcal{O}(N_tN_{\rm RF} N_s + N_{\rm RF}N_s^2)$.
To summarize, the computational complexity of Algorithm~\ref{alg2} is $\mathcal{O}(I_0(N_tN_{\rm RF}N_s +  N_{\rm RF}N_s^2))$.

\section{Simulation Results}\label{sec:6}
\begin{figure*}
	\centering
	\subfigure[$N_r \times N_t = 2\times 2$]{
		\begin{minipage}[t]{0.33\linewidth}
			\centering
			\includegraphics[width=2.3in]{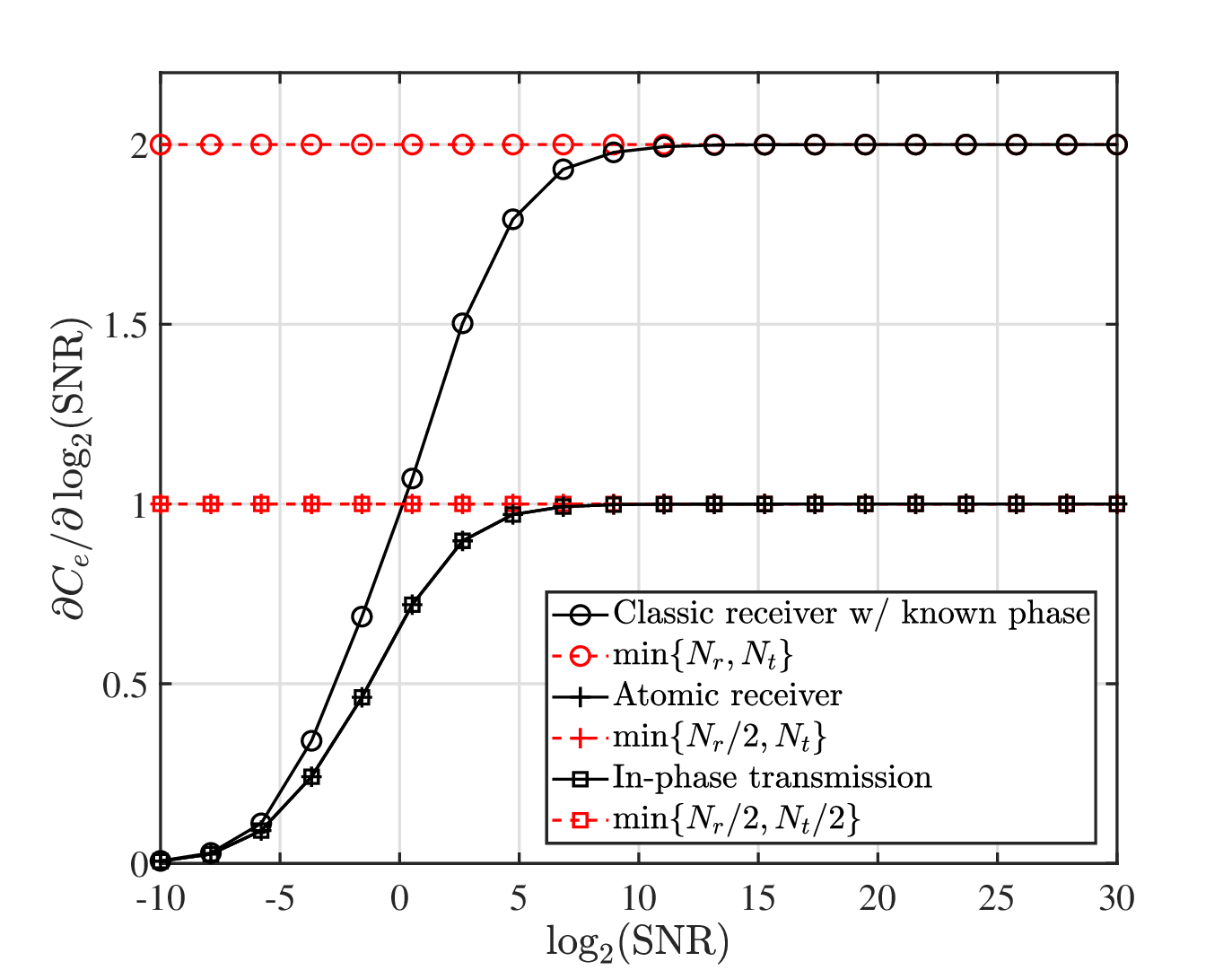}
		\end{minipage}%
	}%
	\subfigure[$N_r \times N_t = 3\times 2$]{
		\begin{minipage}[t]{0.33\linewidth}
			\centering
			\includegraphics[width=2.3in]{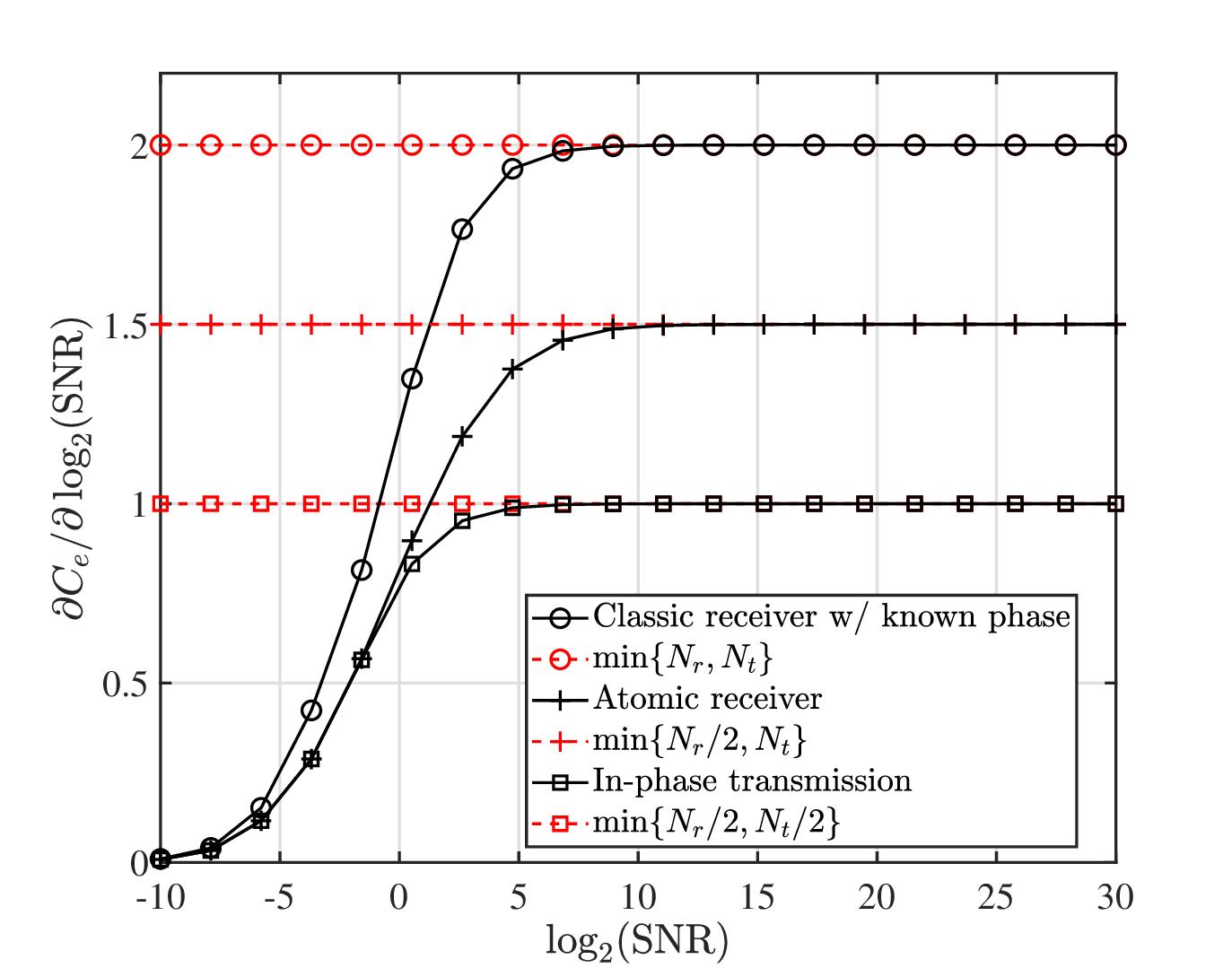}
		\end{minipage}%
	}%
	\subfigure[$N_r \times N_t = 4\times 2$]{
		\begin{minipage}[t]{0.33\linewidth}
			\centering
			\includegraphics[width=2.3in]{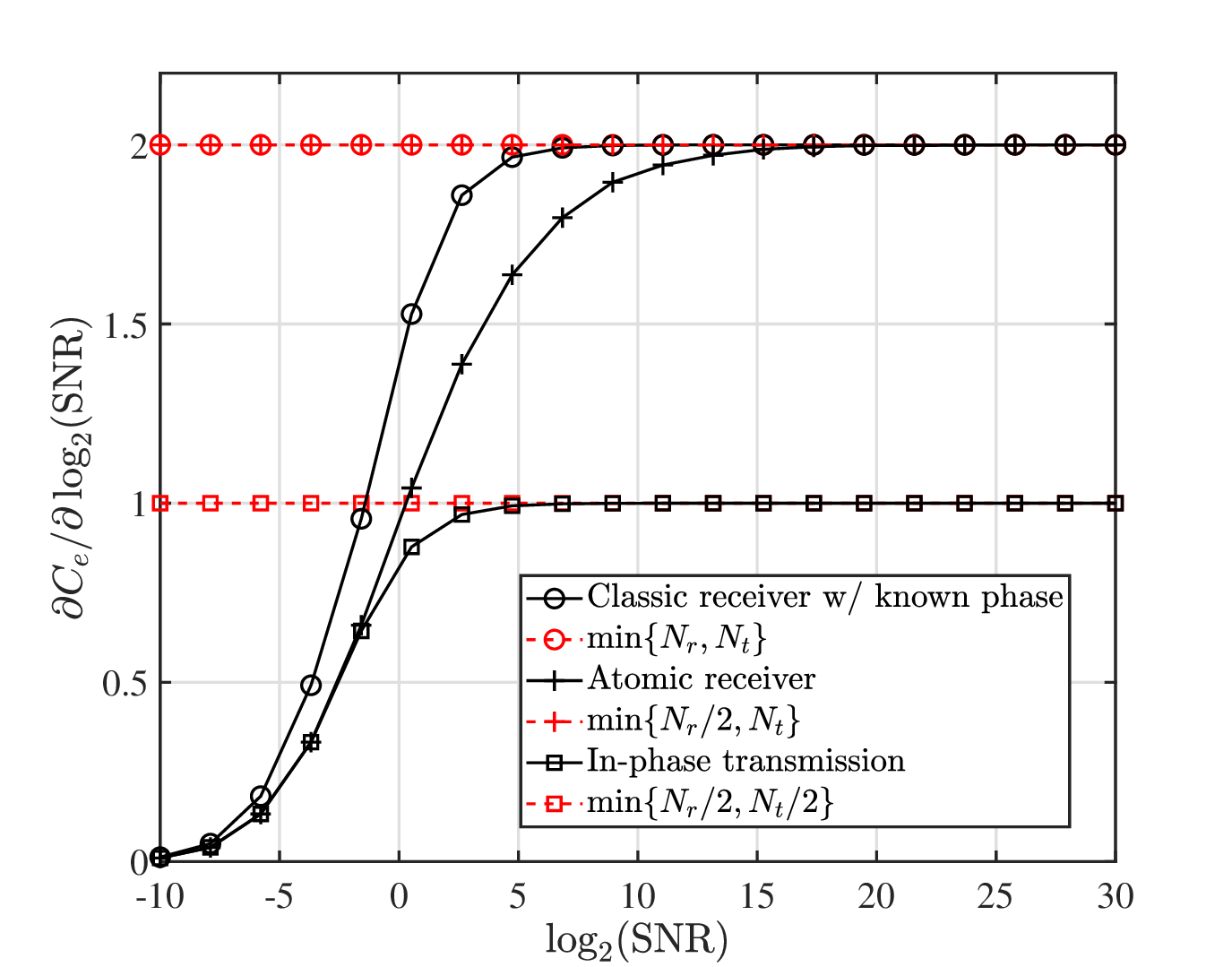}
		\end{minipage}%
	}%
	\centering
	\caption{The derivative of the achievable rate $C$ w.r.t $\log\mathrm{SNR}$ under the (a) $2\times 2$ MIMO, (b) $3\times 2$ MIMO, and (c) $4\times 2$ MIMO configurations.}
	\vspace{-0.2cm}
	\label{img:cl}
\end{figure*}

\subsection{Simulation Settings}
The default simulation settings are as follows. The Rydberg energy levels are set as $62D_{5/2}$ and $64P_{3/2}$ to measure $\omega = 2\pi \times 27.7\:{\rm GHz}$ radio frequency signals. The wavelength is $\lambda = \frac{2\pi c}{\omega} = 10.83\:{\rm mm}$. 
Utilizing the Python package~\cite{SIBALIC2017319}, the transition dipole moment $\boldsymbol{\mu}_{eg}$ over the states $62D_{5/2}$ and $64P_{3/2}$ is calculated as $[0, 789.107\:qa_0, 0]^T$,
where $a_0 = 5.292 \times 10^{-11}$~m specifies the Bohr radius, and $q = 1.602\times 10^{-19}\:{\rm C}$ is the charge of an electron. 
% According to the analytical result in \cite{RydReview_Saffman2010}, the product between the dipole moment and the polarization direction, $\mb{\mu}^T_{\rm RF}\boldsymbol{\epsilon}_{mnl}$, is on the order of $(0.5 \sim 1.5)qn^2a_0$, wherein $n = 52$, $a_0 = 5.292\times 10^{-11}\:{\rm m}$, and $q = 1.602\times 10^{-19}\:{\rm C}$ is the charge of an electron. 
% Taking into account the randomness of the polarization direction, vectors $\boldsymbol{\epsilon}_{nkl}$ and $\boldsymbol{\epsilon}_{b,n}$ are randomly sampled from a unit sphere. 
The uniform linear arrays (ULAs) are deployed at the transceiver with an equal antenna spacing $d = 10\:{\rm mm}$. The transmitter-to-receiver channel coefficients are generated using the standard Saleh-Valenzuela multi-path channel model. 
Specifically, the number of paths is set to $L = 10$. For each path, the complex coefficient $\rho_{mnl}e^{j\phi_{mnl}}$ is generated by the equation $\rho_{mnl}e^{j\phi_{mnl}} = \alpha_l e^{j\frac{d}{\lambda}(m\sin\vartheta_l + n\sin\varphi_l)}$ according to the ULA geometry. Here, the path gain $\alpha_l$ is sampled from the complex Gaussian distribution $\mathcal{CN}(0, 1)$. The incident and departure angles, $\vartheta_l$  and  $\varphi_l$, are sampled from the uniform distribution $\mathcal{U}(-90^\circ, 90^\circ)$. Besides, the polarization direction $\boldsymbol{\epsilon}_{nkl}$ is randomly sampled on a unit circle perpendicular to the incident angle.  
As for the LO-to-receiver channel, since the LO-to-receiver distance is quite small, we fix the path losses $\{\rho_{R,m}\}$ as $\rho_R$ and generate the phases $\{\phi_{R,m}\}$ from the distribution $\mathcal{U}(0, 2\pi)$.
The \emph{SNR}  is defined as $\frac{P}{\sigma^2}$, while the \emph{receive SNR} is defined as 
\begin{align}
    {\rm Receive\:\:SNR} = \frac{\mathsf{E}(\|{\mb{H}}\mb{x}\|_2^2)}{\mathsf{E}(\|\mb{n}\|_2^2)}.
\end{align}
To quantify the relative strength of the reference signal $\mb{r}$, we introduce the \emph{reference-to-signal-and-noise ratio} (RSNR) defined as 
\begin{align}
    {\rm RSNR} = \frac{\mathsf{E}(\| \mb{r}\|_2^2)}{\mathsf{E}(\|{\mb{H}}\mb{x} + \mb{n}\|_2^2)}.
\end{align}
In our simulations, the parameters, $\rho_R$, $P$, and $\sigma^2$, are properly adjusted to increase the receive SNR from -10 dB to 5 dB and the RSNR from -5 dB to 25 dB accordingly.

\subsection{DoFs of Atomic Receivers}

% \begin{figure*}
% 	\centering
% 	\subfigure[$N_r \times N_t = 2\times 2$]{
% 		\begin{minipage}[t]{0.33\linewidth}
% 			\centering
% 			\includegraphics[width=2.5in]{Figures/DoF_SNR_2_2.eps}
% 		\end{minipage}%
% 	}%
% 	\subfigure[$N_r \times N_t = 3\times 2$]{
% 		\begin{minipage}[t]{0.33\linewidth}
% 			\centering
% 			\includegraphics[width=2.5in]{Figures/DoF_SNR_3_2.eps}
% 		\end{minipage}%
% 	}%
% 	\subfigure[$N_r \times N_t = 4\times 2$]{
% 		\begin{minipage}[t]{0.33\linewidth}
% 			\centering
% 			\includegraphics[width=2.5in]{Figures/DoF_SNR_4_2.eps}
% 		\end{minipage}%
% 	}%
% 	\centering
% 	\caption{The derivative of the achievable rate $C$ w.r.t $\log\mathrm{SNR}$ under the (a) $2\times 2$ MIMO, (b) $3\times 2$ MIMO, and (c) $4\times 2$ MIMO configurations.}
% 	\vspace{-0.2cm}
% 	\label{img:cl}
% \end{figure*}
To begin with, the DoFs of atomic MIMO receivers under high SNR regimes are evaluated in Fig.~\ref{img:cl}. We conduct 1000 Monte-Carlo trails to simulate the derivative of the channel capacity $C$ w.r.t. $\log{\rm SNR}$. The plotted curves include: 1) the simulated DoFs of a classic receiver when the phase of $\mb{H}\mb{x} + \mb{r} + \mb{w}$ is known, whose DoFs are theoretically $\min(N_t, N_r)$; 2) the simulated DoFs of an atomic receiver with theoretical DoFs of $\min(N_t/2, N_r)$; and 3) the simulated DoFs of in-phase transmission, where 
 only the in-phase symbols, $\mb{s}_I$, are transmitted and its theoretical DoFs are $\min(N_t/2, N_r/2)$. 
 The number of the transmit antennas is fixed as $N_t = 2$, while the number of receive antennas, $N_r$, grows from $2$ to $4$. 
 We can observe from Fig.~\ref{img:cl} that, with the growth of $N_r$, the DoFs of an atomic receiver gradually increase from 1 to 2. Thus, its capacity limit resembles the in-phase transmission when $N_r = N_t = 2$ and becomes equivalent to the classic receiver when $N_r = 2N_t = 4$. This observation is consistent with our analytical results in Section~\ref{sec:3}.  
 Even if an atomic receiver is hard to measure the phase of incident EM waves, it experiences no loss in DoF with the help of a strong reference signal and sufficient number of receive antennas.
 By further considering that the atomic receiver has a stronger sensitivity and thus has a much higher receive SNR, 
 it has the potential to outperform the classic counterpart. 

\subsection{Validation of Strong Reference Approximation}
\begin{figure}
    \centering
    \includegraphics[width=3.5in]{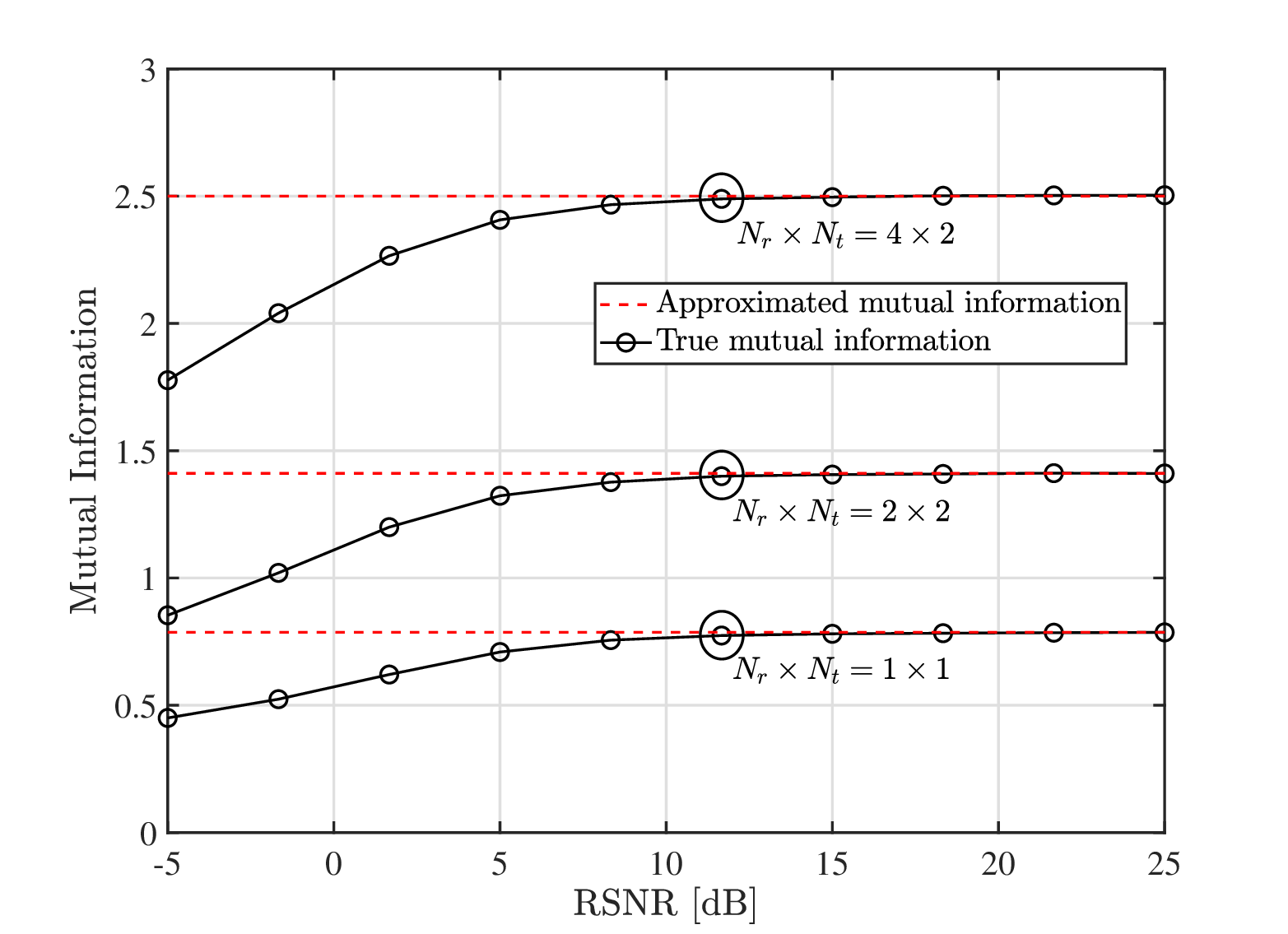}
    % \vspace*{-1em}
    \caption{Mutual information between $\mb{y}$ and $\mb{x}$ against RSNR for the magnitude-detection model in \eqref{eq:model} and real-part-detection model in \eqref{eq:real}.} %图片标题
	% \vspace*{-1em}
	\label{img:MI_RSNR}
\end{figure}
To validate the accuracy of the strong-reference approximation, we plot the curves of the approximated mutual information, $\mathcal{I}(\Re(\widetilde{\mb{H}}\mb{x}) + \mb{w}_I; \mb{x})$, and the true mutual information, $\mathcal{I}(\mb{y}; \mb{x})$, derived from the nonlinear model \eqref{eq:model}. As presented in Fig.~\ref{img:MI_RSNR}, three transmitter and receiver configurations are considered: $N_r \times N_t = 1\times 1$,  $2\times 2$, and $4\times 2$. The RSNR grows from $-5$ dB to $25$ dB, while the receive SNR is fixed as 0 dB. Owing to the nonlinear input-output relationship, $\mathcal{I}(\mb{y}; \mb{x})$ has no analytical expressions. Furthermore, calculating $\mathcal{I}(\mb{y}; \mb{x})$ numerically involves the high-dimensional integration of noncentral chi-square distribution, which is computationally prohibitive in practice. To address this issue, we adopt the Monte-Carlo method proposed in \cite{nonlinear_psaltopoulos2010} to obtain $\mathcal{I}(\mb{y}; \mb{x})$ numerically. As illustrated in Fig.~\ref{img:MI_RSNR}, the true mutual information gradually converges to the approximated mutual information as RSNR increases. When RSNR is greater than 10 dB, the relative error of mutual information is no more than 1\%, which is accurate enough for practical approximation. 
Given that the LO-to-atomic-receiver distance is at the centimeter level, which is dozens of times shorter than the transmitter-to-atomic-receiver distance at the meter level, achieving a 10-dB RSNR is feasible in practical scenarios. This validates the effectiveness of the strong-reference approximation. 

\subsection{Achievable Rates of Proposed IQ-aware Precoding}
\begin{figure}
    \centering
    \subfigure[Achievable rate vs receive SNR]{\includegraphics[width=3.5in]{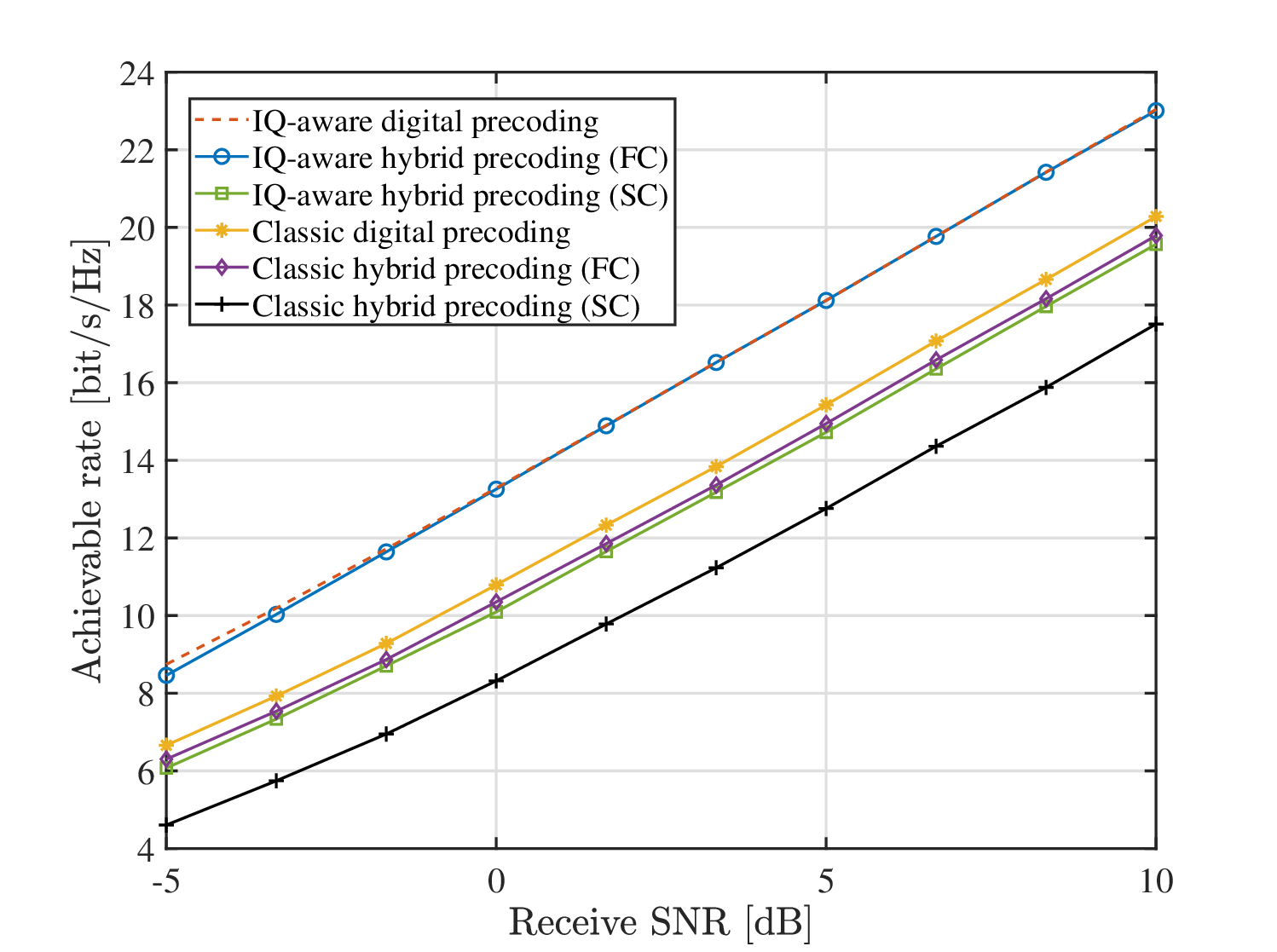}}\\
    % \vspace*{-1em}	
    \subfigure[Achievable rate vs number of receive antennas]{\includegraphics[width=3.5in]{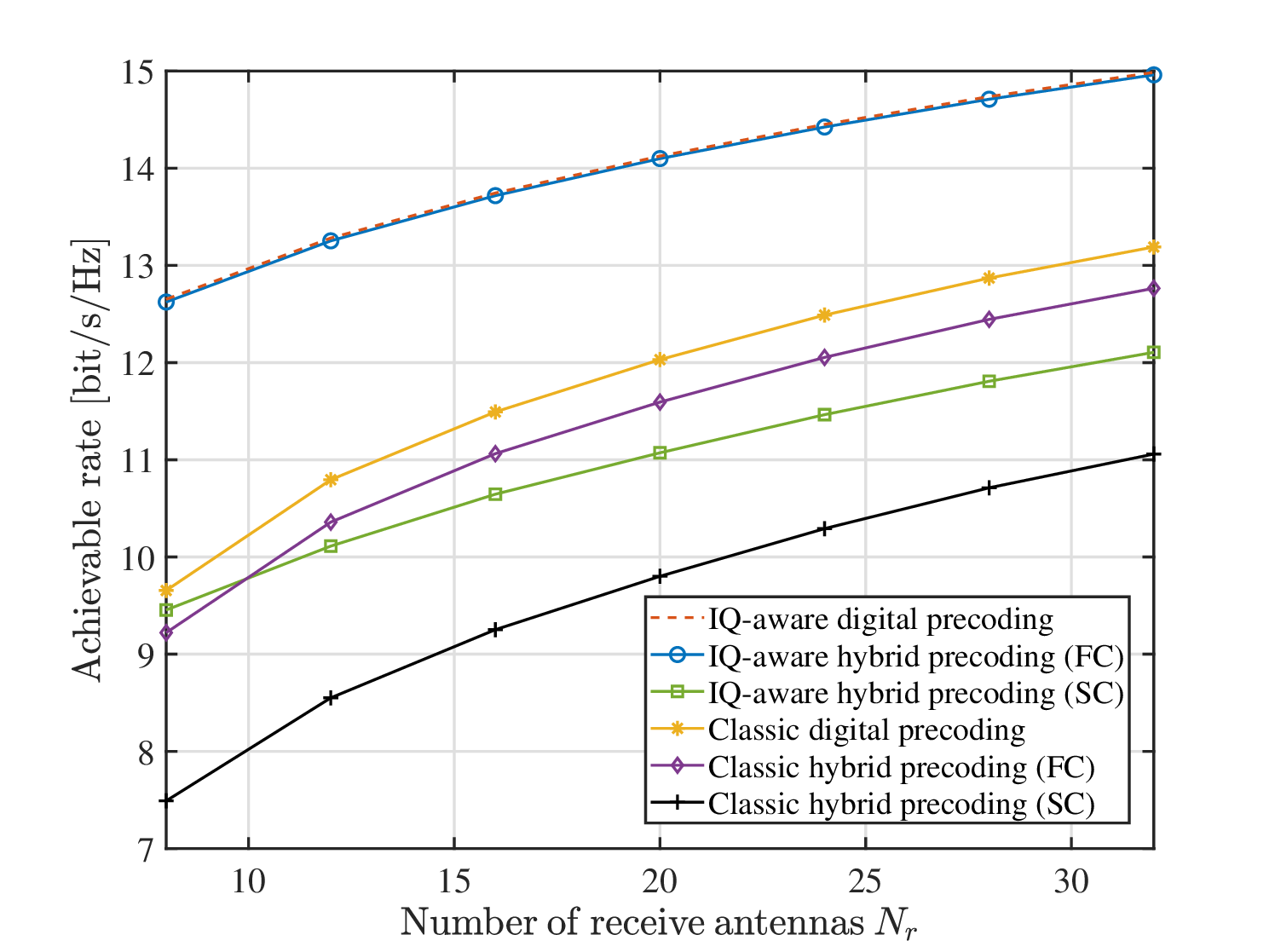}} 
    % \vspace*{-1em}
    \caption{The achievable rate [bits/s/Hz] as a function of (a) the receive SNR [dB] and (b) the number of receive antennas, $N_r$.} %图片标题
	\vspace*{-1em}
	\label{img:rate}
\end{figure}
The curves in Fig.~\ref{img:rate} compare different transmitter precoding approaches. The achievable rate of atomic MIMO receivers in \eqref{eq:MI} is adopted for performance evaluation. 
The benchmarks include: 1) the optimal IQ-aware digital precoding;  2)-3) the proposed IQ-aware hybrid precoding for the FC- and SC-PS structures; 4) the classic digital precoding, which assigns the complex precoder $\mb{F}$ with the principal singular vectors of the complex channel $\widetilde{\mb{H}}$ associated with water-filling principle; 5)-6) the classic hybrid precoding algorithms, PE-AltMin for the FC structure and SDR-AltMin algorithm for the SC structure~\cite{yu_alternating_2016}. For both sub-figures, the number of data streams, RF chains, and transmit antennas are set to $N_s = 3$, $N_{\rm RF} = 12$, and $N_t = 48$. In Fig.~\ref{img:tx} (a), the number of receive antennas is $N_r = 12$ and the receive SNR grows from -5~dB to 10~dB, while in Fig.~\ref{img:tx} (b), the receive SNR is fixed as 0 dB and $N_r$ increases from 8 to 32. We run 1000 Monte-Carlo trails to plot each curve.

We can observe that all proposed IQ-aware precoding schemes outperform their traditional counterparts for all considered SNRs and the number of receive antennas. For example, when ${\rm SNR} = 0$ dB and $N_r = 12$, the gap in achievable rate is approximately 3 bit/s/Hz for the FC hybrid precoding architecture and around 2 bit/s/Hz for the SC hybrid precoding architecture, demonstrating the effectiveness of the proposed design.  
One can also observe that with the increase in SNR, the rate achieved by IQ-aware hybrid precoding (FC) gradually approaches that of the optimal IQ-aware digital precoding. This comes from the fact that the power allocation to the orthogonal columns of the digital precoder $\bar{\mb{F}}$ tend to be uniform at high SNRs, making the orthogonal-column assumption applied on $\bar{\mb{D}}$ more precise. Consequently, we can conclude that an analog FC-PS network associated with a low-dimensional IQ-aware digital precoder is able to accurately approximate a high--dimensional IQ-aware digital precoder. 

% although the analog circuits cannot independent process the IQ symbols, 

% We can observe that both the proposed WMMSE and IQ-WMMSE algorithms can approach the capacity realized by IQ-aware precoding under high-SNR and large-receiver-antenna-number regimes. Particularly, when $N_r = 32$ in Fig.~\ref{img:tx} (b), the achievable rate is improved by 3 bits compared to the classic SVD precoding. Another interesting finding is that the WMMSE algorithm performs worse than single-channel precoding under low-SNR and small-receiver-antenna-number regimes. This is because the WMMSE algorithm lacks the ability to distribute non-isotropic powers to IQ channels, resulting in power leakage to unwanted data streams. However, with the introduced IQ power allocator, the proposed IQ-WMMSE overcomes this issue and outperforms other benchmarks in all considered SNR and antenna-number conditions. For example, when $N_r = 32$ and ${\rm SNR} = 10$ dB, the performance gap between IQ-WMMSE and WMMSE is around 1 bit. This demonstrates the effectiveness of the proposed designs. 

\subsection{Convergence Verification}
\begin{figure}
    \centering
    \subfigure[IQ-aware FC hybrid precoding]{\includegraphics[width=3.5in]{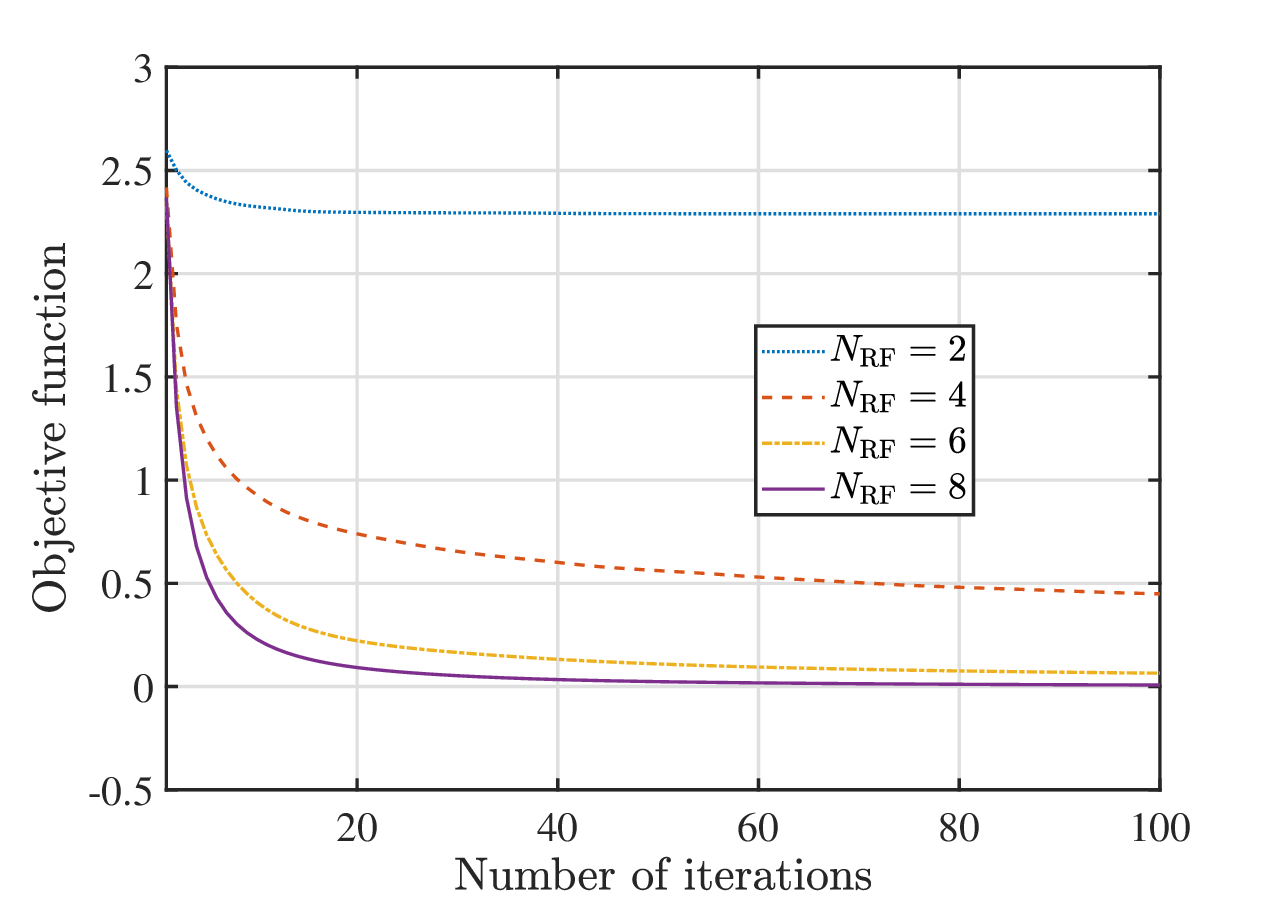}}\\
    % \vspace*{-1em}	
    \subfigure[IQ-aware SC hybrid precoding]{\includegraphics[width=3.5in]{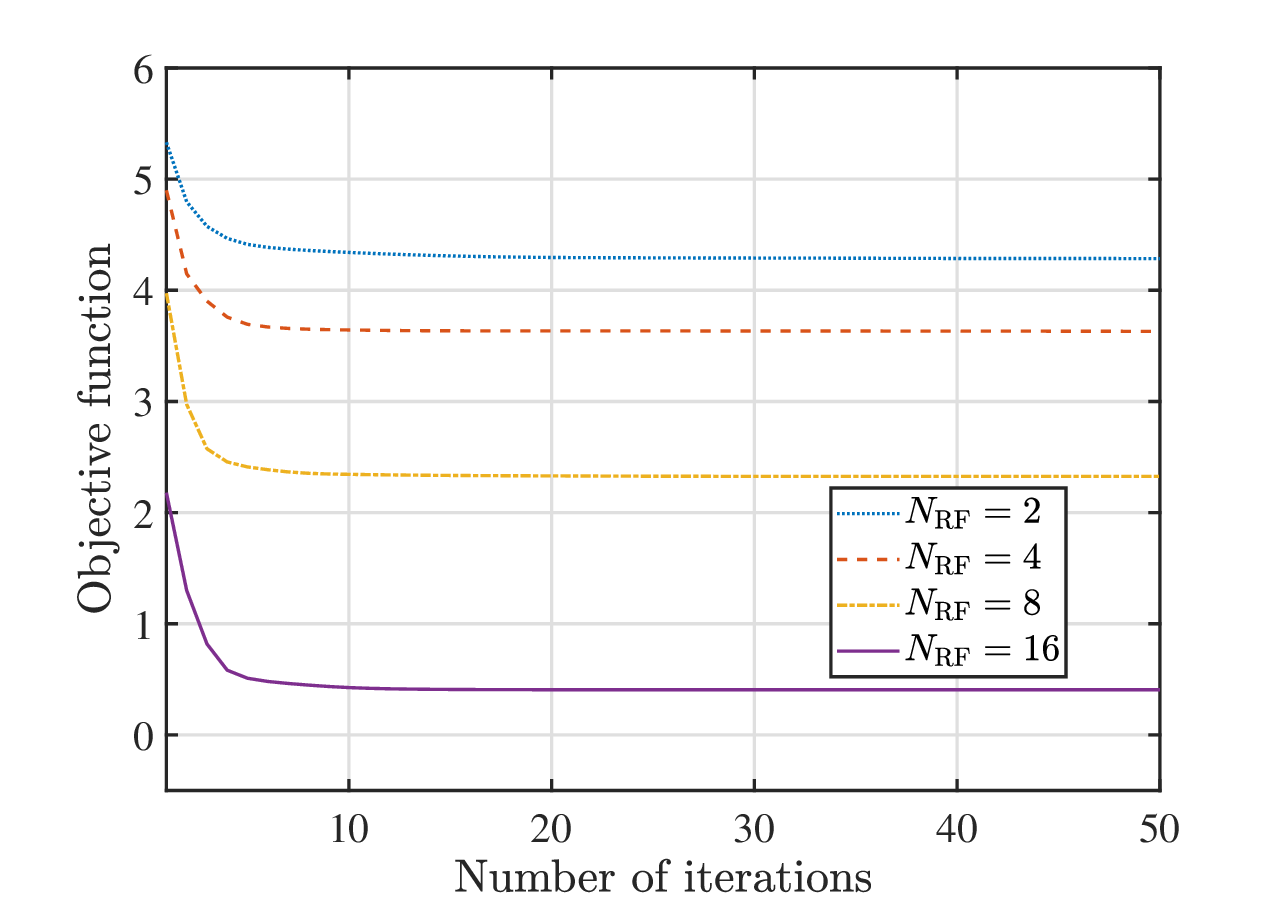}} 
    % \vspace*{-1em}
    \caption{Objective functions, \eqref{eq:FCP2} and \eqref{eq:SCP1}, w.r.t the number of iterations.} %图片标题
	\vspace*{-1em}
	\label{img:rate_iter}
\end{figure}
% \begin{figure}
%     \centering
%     \includegraphics[width=3.5in]{Figures/Rate_iter.eps}
%     \vspace*{-1em}	
%     \caption{The achievable rate w.r.t the number of iterations under the $4\times 2$, $4\times 4$, and $8\times 4$ MIMO configurations. } %图片标题
% 	\vspace*{-1em}
% 	\label{img:rate_iter}
% \end{figure}

The curves in Fig.~\ref{img:rate_iter} demonstrate the convergence of Algorithm 1 and Algorithm 2. The value of objective functions, $\| \widetilde{\mb{F}}\widetilde{\mb{D}}^T -  \bar{\mb{A}}\|_F^2$ and $\| \bar{\mb{F}} -  \bar{\mb{A}}\bar{\mb{D}}\|_F^2$, in each iteration are plotted in Fig.~\ref{img:rate_iter}(a) and (b), respectively. The simulation parameters are as follows: $N_s = 2$, $N_t = 32$, $N_r = 12$, and ${\rm SNR} = 0$ dB. It is clear that our proposed algorithms can monotonically decrease the objective functions until convergence. The number of iterations required for convergence is around 50 for Algorithm~1 and is around 10 for Algorithm~2. The larger iteration number of Algorithm~1 comes from the updating of auxiliary variables. Moreover, since the FC structure has much more PSs, the objective function of Algorithm 1 can converge to 0 when $N_{\rm RF} = 8$, while the converged value of Algorithm 2 is larger than 0 when $N_{\rm RF} = 16$.

\section{Conclusions}\label{sec:7}

% This paper investigated the impact of transmitter precoding on the capacity of Rydberg atomic receivers. 
In this paper, we explored precoding design for achieving the capacity of Rydberg atomic receivers in MIMO systems. 
We proposed a strong-reference approximation to linearize the magnitude detection model into a real-part detection model. An IQ-aware digital precoding is accordingly designed to show that the channel capacity of atomic MIMO receiver scales as $\min(N_r/2, N_t)\log(\rm SNR) + \mathcal{O}(1)$, revealing a DoF-lossless transmission when $N_r \ge 2N_t$.
To facilitate the implementation of analog RF circuits, we further delved into the IQ-aware hybrid precoding algorithms for both FC and SC architectures. Simulation results confirmed that the proposed IQ-aware hybrid precoding can tightly approach the optimal IQ-aware digital precoding.

This work presents a significant stride towards the transmitter-side 
signal processing for enhancing atomic MIMO receivers. We conclude by outlining potential avenues for future research. Firstly, this study was limited to point-to-point MIMO systems. In the context of multi-user communications where each user utilizes an atomic receiver, it becomes crucial include inter-user interference into the IQ-aware precoding design. Secondly, our analyses and designs were built on the strong-reference approximation. In case the LO is absent or the strong-reference approximation is invalid, maximizing the mutual information of nonlinear transition model becomes necessary, which remains an open problem. Last but not the least, precoding design for the general heterodyne measurement, where the frequency of LO slightly differs from that of signal~\cite{Rydphase_Holloway2019}, is an unexplored area as well.

% Some potential 
% research directions are summarized as follows:  
% %issues warrant follow-up studies,
% \begin{itemize}
% 	\item \textbf{Transmitter precoding for multi-user atomic receivers}: This paper focused on point-to-point MIMO systems. However, when it comes to multi-user communications where each user is equipped with an atomic receiver, it is significant to consider inter-user interference. In this context, the proposed precoding schemes might not be optimal.
% 	\item \textbf{Precoding design without the strong reference approximation}:
%     Our analyses and designs were built on the strong reference approximation. In case the LO is absent or the strong reference approximation is invalid, maximizing throughput for the nonlinear transition model becomes necessary. This is still an open problem that requires further investigation. 
% \end{itemize}

\appendix
\subsection{Proof of Lemma \ref{prop:sra}}\label{appendix:sra}
The $m$-th entry of the received signal is given as
    % \begin{align}
    %     &y_m = |r_m + \mb{H}(m, :)\mb{x}+ w_m| 
    %     \notag \\&= |r_m|\sqrt{1 + \frac{|\mb{H}(m, :)\mb{x}+ w_m|^2}{|r_m|^2}  + 2\mathrm{Re}\left(\frac{\mb{H}(m, :)\mb{x} + w_m}{r_m}\right)} \notag\\
    %     &\overset{(a)}{\approx} |r_m|\left( 1 + \mathrm{Re}\left(
    %     \frac{\mb{H}(m, :)\mb{x} + w_m}{r_m}
    %     \right)+ \frac{|\mb{H}(m, :)\mb{x}+ w_m|^2}{2|r_m|^2}
    %     \right) 
    %     \notag \\
    %     & \overset{(b)}{=} |r_m| + \mathrm{Re}(e^{-j\angle r_n}\mb{H}(m, :)\mb{x}) + \mathrm{Re}(e^{-j\angle r_m} w_m),
    % \end{align}
        \begin{align}
         &y_m= |r_m + \mb{H}(m, :)\mb{x}+ w_m| 
        \notag \\&= |r_m|\sqrt{1 + \frac{|\mb{H}(m, :)\mb{x}+ w_m|^2}{|r_m|^2}  + 2\mathrm{Re}\left(\frac{\mb{H}(m, :)\mb{x} + w_m}{r_m}\right)} \notag\\
        &\overset{(a)}{\approx} |r_m|\left( 1 + \mathrm{Re}\left(
        \frac{\mb{H}(m, :)\mb{x} + w_m}{r_m}
        \right)+ \frac{|\mb{H}(m, :)\mb{x}+ w_m|^2}{2|r_m|^2}
        \right) 
        \notag \\
        & \overset{(b)}{=} |r_m| + \mathrm{Re}(e^{-j\angle r_n}\mb{H}(m, :)\mb{x}) + \mathrm{Re}(e^{-j\angle r_m} w_m),
    \end{align}
where (a) comes from the Tailor expansion $\sqrt{1 + x} = 1 + \frac{1}{2}x + \mathcal{O}(x^2)$ and (b) holds when $|r_m| \gg |\mb{H}(m, :)\mb{x} + w_m|$. Moreover, since $\omega_n$ is circularly symmetric distributed, the rotation $e^{-j\angle r_m}$ does not affect the distribution of $\omega_m$. Thereby, $\mathrm{Re}(e^{-j\angle r_m} w_m)$ shares the same distribution with the real part of $\omega_m$, i.e.,  $\mathcal{N}(0, \sigma^2/2)$. This completes the proof.

\subsection{Proof of Lemma \ref{lemma1}}
The optimal solution of \eqref{eq:UOP} is expressed as 
% $\| \widetilde{\mb{F}} - \gamma \bar{\mb{A}}\widetilde{\mb{D}}\|_F^2$ as 
\begin{align}
   &\:\:\:\min_{\bar{\mb{A}}, \bar{\mb{D}}, \bar{\mb{F}}_c}\| (\bar{\mb{F}}, \bar{\mb{F}}_c) - \gamma \bar{\mb{A}}(\bar{\mb{D}}_u, \bar{\mb{D}}_c)\|_F^2 \notag \\
   &= \min_{\bar{\mb{A}}}\left( \min_{\bar{\mb{D}}_u} \| \bar{\mb{F}} - \gamma \bar{\mb{A}}\bar{\mb{D}}_u\|_F^2 + \min_{\bar{\mb{F}}_c, \bar{\mb{D}}_c} \| \bar{\mb{F}}_c - \gamma \bar{\mb{A}}\bar{\mb{D}}_c\|_F^2 
   \right) \notag \\ 
   & \overset{(a)}{=} \min_{\bar{\mb{A}}, \bar{\mb{D}}_u} \| \bar{\mb{F}} - \gamma \bar{\mb{A}}\bar{\mb{D}}_u\|_F^2,  
\end{align}
where (a) holds by substituting $\bar{\mb{F}}_c = \gamma \bar{\mb{A}}\bar{\mb{D}}_c$ into $\| \bar{\mb{F}}_c - \gamma \bar{\mb{A}}\bar{\mb{D}}_c\|_F^2$. This completes the proof.

% \balance
\bibliographystyle{IEEEtran}
\bibliography{Reference.bib}

\end{document}